\documentclass[a4paper,12pt,final]{scrartcl}

\usepackage[utf8]{inputenc}

\usepackage[left=2.5cm,right=2.5cm,
       top=1cm,bottom=1cm, includeheadfoot]{geometry}

\usepackage[english]{babel}
\usepackage{amssymb, amsmath, amstext, amsthm, amsfonts}
\numberwithin{equation}{section}
\usepackage{mathtools,relsize}

\usepackage{mathrsfs}

\usepackage[ddmmyyyy]{datetime}

\usepackage[notref,notcite,color]{showkeys}

\usepackage{pst-node}
\usepackage{tikz-cd}

\usepackage{color}

\usepackage{tikz}

\usepackage{tkz-graph}
\GraphInit[vstyle = Shade]
\tikzset{
  LabelStyle/.style = {rectangle, rounded corners, draw,
                        minimum width = 2em, minimum height =2em ,
                        fill = black!5, text = black },
  VertexStyle/.style = {circle, draw, fill=black!20, 
                        minimum width= 2.2em, font =  \bfseries},
  EdgeStyle/.style = {->, bend left} }

\makeatletter
\renewcommand*\env@cases[1][1.2]{%
  \let\@ifnextchar\new@ifnextchar
  \left\lbrace
  \def\arraystretch{#1}%
  \array{@{\,}c@{\quad}l@{}}%
}
\makeatother

\newcommand{\N}{\mathbb{N}}
\newcommand{\R}{\mathbb{R}}
\newcommand{\Rinfty}{\R_\infty}
\newcommand{\eps}{\varepsilon}

\newcommand{\sfC}{\mathsf{C}} 
\newcommand{\sfQ}{\mathsf{Q}} 

\newcommand{\cB}{\mathcal{B}}

\newcommand{\cE}{\mathcal{E}}
\newcommand{\cI}{\mathcal{I}}
\newcommand{\cJ}{\mathcal{J}}
\newcommand{\cR}{\mathcal{R}}

\newcommand{\cD}{\mathcal{D}}

\newcommand{\fD}{\mathfrak D}
\newcommand{\mfB}{\mathfrak B}
\newcommand{\rmC}{\mathrm{C}}
\newcommand{\rmL}{\mathrm{L}}
\newcommand{\LB}{\lambda_{\mathrm{Bz}}}
\newcommand{\rmW}{\mathrm{W}}

\newcommand{\bigset}[2]{\big\{\, #1 \,\big| \, #2 \,\big\} }
\newcommand{\ti}{{\times}}
\newcommand{\eff}{\mathrm{eff}}
\newcommand{\wt}{\widetilde}
\newcommand{\la}{\langle}
\newcommand{\ra}{\rangle}
\newcommand{\DDD}{\mathbb D}
\newcommand{\Gammlim}{\xrightarrow{\text{$\Gamma$}}}
\newcommand{\wGammlim}{\stackrel{\text{$\Gamma$}}{\rightharpoonup}}
\newcommand{\Moscolim}{\xrightarrow{\mathrm{M}}}
\newcommand{\EDPto}{\xrightarrow{\,\mathrm{EDP}\,}}

\newcommand{\olka}{\overline{\kappa}}%
\newcommand{\ulw}{\underline{w}}%
\newcommand{\ulk}{\underline{\kappa}}%

\newcommand{\D}{\mathrm{D}}
\renewcommand{\d}{\mathrm{d}}
\newcommand{\dd}{\hspace{0.25ex}\mathrm{d}\;\!} 

\definecolor{rot}{rgb}{1.000,0.000,0.000}

\newcommand{\one}{1\!\!1}
\newcommand{\e}{\mathrm{e}}

\newtheorem{thm}{Theorem}[section]

\newtheorem{lem}[thm]{Lemma}
\newtheorem{prop}[thm]{Proposition}

\theoremstyle{definition}
\newtheorem{defi}[thm]{Definition}
\newtheorem{rem}[thm]{Remark}
\newtheorem{exam}[thm]{Example}

\newcommand{\AAAERR}{\color{blue}}
\newcommand{\EEE}{\color{black}}

\usepackage{fancyhdr}
\title{Coarse-graining via EDP-convergence 
  for linear fast-slow reaction systems%
    \thanks{This research has been funded by 
         Deutsche Forschungsgemeinschaft (DFG) through SFB\,1114 
         ``\emph{Scaling Cascades in Complex Systems}'', 
         Project C05 ``Effective Models for Materials and Interfaces with
         Multiple Scales''.}
}

\author{Alexander Mielke%
    \thanks{Weierstraß-Institut für Angewandte
      Analysis und Stochastik, Mohrenstraße 39, 10117 Berlin
      and Humboldt-Universität zu Berlin, Unter den Linden 6, 10099
      Berlin, Germany}\ \ 
and Artur Stephan%
    \thanks{Weierstraß-Institut für Angewandte Analysis und Stochastik, 
        Mohrenstraße 39, 10117 Berlin}}

\date{\normalsize 29 May 2020, \AAAERR with erratum of 16 June 2021.\footnotemark[4]}

\begin{document}

\maketitle

\footnotetext{\AAAERR A previous version of this work stated Lemma 3.4 with too weak
  assumptions.}

\begin{abstract}
  We consider linear reaction systems with slow and fast reactions,
  which can be interpreted as master equations or Kolmogorov
  forward equations for Markov processes on a finite state space.
  We investigate their limit behavior if the fast reaction
  rates tend to infinity, which leads to a coarse-grained
  model where the fast reactions create microscopically equilibrated
  clusters, while the exchange mass between the clusters occurs 
  on the slow time scale. 

  Assuming detailed balance the reaction system can be written as a
  gradient flow with respect to the relative entropy. Focusing on
  the physically relevant cosh-type gradient structure we show
  how an effective limit gradient structure can be rigorously derived
  and that the coarse-grained equation again has a cosh-type gradient
  structure. We obtain the strongest version of convergence in the
  sense of the Energy-Dissipation Principle (EDP), namely
  EDP-convergence with tilting.
\end{abstract}




\section{Introduction}

Considering $I\in\N$ particles that interact linearly with each other
with given rates $A_{ik}$, the evolution of the probability or
concentration $c_i\in [0,1]$ of a species $i\in\{1,\dots, I\}=:
\mathcal I$ can be described by the master equation
\begin{align}\label{eq1}
 \dot c = A c,
\end{align}
where $A$ is the adjoint of the Markov generator $\mathcal L:\R^I\to
\R^I$ of the underlying Markov process, i.e.\ $A=\mathcal L^*$, see
e.g.\ \cite{Dynk65MP, Bobr05FAPS, Durr10PTE} for more
information. In particular, this means $A_{ki}\geq 0$ for $i\neq
k$ and $\sum_{k=1}^I A_{ki}=0$ for all $i\in \cI$.  
We interpret the master equation as a
\textit{rate equation} defined on the state space 
\[
\sfQ=\mathrm{Prob}(\cI) :=\bigset{c\in [0,1]^I}{ \sum\nolimits_{i=1}^I c_i = 1}  
\subset \R^I. 
\]

In many applications the number $I$ of particles can be huge and the
reaction coefficients $A_{ik}$ may vary in a huge range. In such cases
the analysis or the numerical treatment of system \eqref{eq1} 
is out of reach, and hence
suitable simplifications are necessary. One natural assumption is that
reactions can happen with different speeds. We will consider the case
the slow and fast
reactions are distinguished, the slow ones of order 1 and the fast
ones of order $1/\eps$ for a small parameter $\eps\to 0$. Hence, we
decompose $A=A^\eps$ into 
$A^\eps = A^S + \frac{1}{\eps}A^F$, ``$S$'' for slow and ``$F$''
for fast reactions. Our equation then is $\eps$-dependent and reads
\begin{align}\label{eq2}
 \dot c^\eps = A^\eps c^\eps = \big(A^S + \frac 1 \eps A^F\,\big)\;c^\eps.
\end{align}
The limit passage for $\eps\to 0$ in linear and nonlinear
slow-fast reaction systems is a well-established field starting from
pioneering work by Tikhonov \cite{Tikh52SDEC} and Fenichel
\cite{Feni79GSPT}.  We refer to \cite{Both03ILRC, DiLiZi18EGCG,
  Step19EDPCGSDTS} \EEE for modern 
approaches and to \cite{KanKur13STSM} for nonlinear fast-slow reaction
systems under the influence of stochastic
fluctuations, see e.g.\ Example~6.1 there for a mRNA-DNA system for
$I=6$ species with 8 slow reactions and 2 fast reactions.

While we repeat some of these arguments in Section
\ref{se:EpsDependReactNetw}, the main goal of this paper is quite
different. Our study is devoted to \EEE the associated \emph{gradient
  structures} for \eqref{eq2} and their limiting behavior for
$\eps \to 0$. Gradient structures \EEE  exist under the additional
assumption that the \emph{detailed-balance condition} holds, 
means that there exists a positive equilibrium state
$w^\eps=(w^\eps_i)_{i\in \cI} \in \sfQ$ such that
\begin{equation}
  \label{eq:I.DBC}
  \text{detailed-balance condition (DBC):} \qquad \forall\,i,k\in \cI:\ \ 
A^\eps_{ik} w^\eps_k=A^\eps_{ki}w^\eps_i.
\end{equation}
Following \cite{Miel11GSRD,Pele14VMEG,Miel16EGCG}, a gradient structure for a
rate equation $\dot c= V_\eps(c)$ on the state space $\sfQ$ means that
there exist a differentiable energy functional $\cE_\eps$ and a
dissipation potential $\cR_\eps$ such that the rate equation can be
generated as the associated gradient-flow equation, namely
\begin{equation}
 \label{eq:I.GradFlow}
  \dot c = V_\eps(c)= \D_\xi\cR_\eps^*(c,{-}\D \cE_\eps(c))
  \quad \text{ or equivalently }\quad 
  0 = \D_{\dot c} \cR_\eps(c,\dot c)+ \D \cE_\eps(c).
\end{equation}
Here $\cR_\eps$ is called a dissipation potential if
$\cR_\eps(c,\cdot) : \mathrm T_c \sfQ\to [0,\infty]$ is lower
semicontinuous and convex and satisfies $\cR_\eps(c,0)=0$. Then, 
$\cR^*_\eps$ is the (partial) Legendre-Fenchel transform 
\[
\cR^*_\eps(c,\xi):=\sup\bigset{\la \xi,v\ra -\cR_\eps(c,v)}
                      { v \in \mathrm T_c\sfQ }.
\]

For reaction systems of mass-action type (which includes all linear
systems) satisfying detailed balance, it was shown in
\cite{Miel11GSRD} that an entropic gradient structure exists, i.e.\
$\cE_\eps$ is the relative Boltzmann entropy $\cE_\mathrm{Bz}^\eps (c)
:= \mathcal H (c|w^\eps)$ of $c$ with respect to $w^\eps$, see Section
\ref{suu:QuadrEntrGS}. However, this fact was used implicitly in
earlier works, see e.g.\ \cite[Eqn.\,(113)]{OttGrm97DTCF2} and
\cite[Sec.\,VII]{Yong08ICPD}. For linear reaction systems, which are
master equations for Markov processes, a more general theory was
developed in \cite{Maas11GFEF,CHLZ12FPEF} leading to a large class of
possible gradient structures, see Section \ref{SectionGeneralGS} and
\cite[Sec.\,2.5]{MaaMie18?MCRD}.

Here, we use the physically most natural gradient structure that has
its origin in the theory of large deviation, see \cite{MiPeRe14RGFL,
  MPPR17NETP}. The dual dissipation potentials $\cR^*_\eps (c,\cdot):
\mathrm T_c \sfQ\to \R$ are not quadratic but rather exponential due
to cosh terms, namely
\begin{align}
  \label{eq:I.cosh}
\cR^*_\eps(c,\xi) =  \frac12 \sum_{i<k}  \kappa_{ik}^\eps 
\sqrt{c_ic_k}\, 
\sfC^*(\xi_{i} {-} \xi_{k}) \ \text{ with } 
\sfC^*(\zeta)= 4\cosh(\zeta/2)-4
\end{align}
and $\kappa^\eps_{ik}=A^\eps_{ik} \sqrt{w_k^\eps
  /w_i^\eps}$. The gradient structure $(\sfQ,\cE_\mathrm{Bz}^\eps,
\cR^*_\eps)$ exactly generates the gradient-flow evolution
\eqref{eq2}, and we call it simply the \emph{cosh
  gradient structure}. Note that the dissipation potential
$v \mapsto \cR_\eps(c,\cdot)$ is still superlinear, but grows only
like $|v|\log(1{+}|v|)$. In particular, $\cR_\eps$ does not
induce a metric on $\sfQ$. \EEE

This gradient structure is also in line with the first derivation of
exponential kinetic relations by Marcellin in 1915, see
\cite{Marc15CECP}. Moreover, it arises as effective gradient structure
in EDP converging systems, see \cite{LMPR17MOGG, FreLie19?EDTS}. In
\cite{FreMie19?DKRF} it is shown that the exponential function
``$\cosh$'' arises due to the Boltzmann entropy as inverse of the
logarithm. For $\rmL^p$-type entropies $\cR^*$ will have a growth like
$|\xi|^{c_0/(p-1)}$.

Instead of passing to the limit $\eps\to 0$ in the equation
\eqref{eq2}, our goal is to perform the limit passage in the gradient
system $(\sfQ,\cE_\mathrm{Bz}^\eps, \cR^*_\eps)$ to obtain directly an
effective gradient system $(\sfQ,\cE_0, \cR^*_\eff)$ via the notion of
EDP-convergence as introduced in \cite{LMPR17MOGG, DoFrMi19GSWE,
MiMoPe18?EFED}. Roughly spoken this convergence asked for the
$\Gamma$-convergence of the energies, namely $\cE^\eps_\mathrm{Bz}\Gammlim
\cE_0$ on $\sfQ$, and for the dissipation functionals $\fD_\eps \Gammlim
\fD_0$ on $\rmL^2([0,T];\sfQ)$ with
\begin{align*}
 &\fD_\eps(c)= \int_0^T \!\! \Big( \cR_\eps(c,\dot c){+} 
     \cR^*_\eps(c,{-}\D\cE_\eps(c)) \Big) \dd t \ \text{ and }\\  
 &\fD_0(c)= \int_0^T \!\! \Big(\cR_\eff(c,\dot c) {+} 
            \cR^*_\eff(c,{-}\D\cE_0(c))\Big) \dd t.
\end{align*}
The notion of EDP-convergence produces a \emph{unique limit gradient
  system}, and we may have $\cR_\eps \Gammlim \cR_0$ while $\cR_\eff
\neq \cR_0$, see \cite{LMPR17MOGG, DoFrMi19GSWE}.  As a trivial
consequence of EDP-convergence we then find that $0=\D \cR_\eff(c,\dot
c) + \D \cE_0(c)$ is the limit equation, cf.\ Lemma
\ref{le:EDPimpliesCvg}. 

We emphasize that constructing the limit equation $\dot c=
V^0(c)$ for a family of
evolution equations $\dot c=V^\eps(c)$  depending on a small parameter
$\eps\to 0$ is quite 
different from our goal. In Figure \ref{fig:EDP.comm.diagram} this
would mean to concentrate on the two downward arrows on the right
only. In general, an evolution equation may have many gradient
structures, and this is certainly true for the linear master equations
studied here. Thus, from knowing the limiting equation $\dot c =
V^0(c)$ we cannot recover a unique gradient structure
$(\sfQ,\cE_0,\cR_{\mathrm{eff}})$. 

Our philosophy is opposite: We consider the gradient structure
$(\sfQ,\cE_\eps,\cR_\eps)$ associated to $\dot c= V^\eps(c)$ as
additional information that is not contained
in the equation, but of course they are compatible. Typically, the
additional information is of thermodynamical nature and reflects the
underlying microscopic properties of the model that are no longer seen
in the macroscopic model, 
see \cite{MiPeRe14RGFL, MPPR17NETP}. In \cite{PeReVa14LDSH} it is
shown that the parabolic equation $\dot u = u_{xx}$ associates with different
gradient structures if one models diffusion or if one models heat transfer. \EEE

 Thus, we turn around the
usual limit analysis where one first works on the gradient-flow
equations \eqref{eq:I.GradFlow} and the solutions $c^\eps:[0,T]\to
\sfQ$, and then studies gradient structures for the limit
equations. As indicated \EEE in Figure \ref{fig:EDP.comm.diagram},
EDP-convergence works solely on the gradient systems and produces
$\cR_\eff$ as a nontrivial result, which then gives the limit equation
and the accumulation points $c^0:[0,T]\to \sfQ$ of the solutions
$c^\eps:[0,T]\to \sfQ$.
\begin{figure}
\centering
\begin{tikzpicture}[scale=1.1]
\node[fill=black!0] at (0,2.3) {\smaller gradient systems};
\node[fill=black!30] at (0,1.7) {$(\sfQ,\cE_\eps,\cR_\eps)$};

\node[fill=black!30] at (0,0) {$(\sfQ,\cE_0,\cR_{\mathrm{eff}})$};

\node at (1.4,1.7) {$\leadsto$}; 

\node[fill=black!0] at (4.4,2.3) {\smaller gradient-flow eqn.};
\node[fill=black!20] at (4.4,1.7) 
  {$\dot c=\partial_{\xi}\cR^*_\eps(c,{-}\D\cE_\eps(c))=V^\eps(c)$}; 

\node at (7.4,1.7) {$\leadsto$}; 

\node[fill=black!0] at (9,2.3) {\smaller solutions};
\node[fill=black!10] at (9,1.7) {$c^\eps{:}\,[0,T]\to \sfQ$};

\node[fill=red!30, rotate=270] at (-1.6,0.85) {$\eps\to 0$};
\node[fill=red!40,rotate=270] at (0,0.85) 
                  {$\overset{\text{EDP}}{\longrightarrow}$};
\node[fill=red!20,rotate=270] at (6.3,0.85) {$\leadsto$};
\node[fill=red!20, rotate=270] at (9,0.85) {$\rightharpoonup$};

\node at (1.4,0) {$\leadsto$}; 
\node at (7.4,0) {$\leadsto$}; 

\node[fill=black!20] at (4.4,0) 
     {$\dot c=\partial_{\xi}\cR^*_{\mathrm{eff}}(c,{-}\D\cE_0(c))=V^0(c)$}; 
\node[fill=black!10] at (9,0) {$c^0{:}\,[0,T]\to \sfQ$};

\end{tikzpicture}%

\caption{EDP-convergence leads to a commuting diagram, in particular
  EDP-convergence generates the correct limit equation $\dot c =
  V^0(c)$ and (subsequences of) the solutions $c^\eps$ converge to
  solutions $c^0$ of the limit equation. However, $\cR_\eff$
  provides information not contained in the limit equation.}
\label{fig:EDP.comm.diagram}
\end{figure}
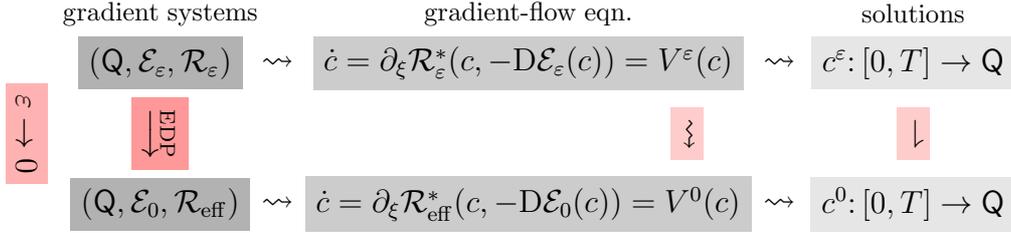

The choice of gradient structure for a given family of equations 
$\dot c=V^\eps(c)$ may even be relevant for deriving the effective of
limiting equation. Choosing the gradient structures
$(\sfQ,\cE_\eps^{(1)},\cR_\eps^{(1)})$ and
$(\sfQ,\cE_\eps^{(1)},\cR_\eps^{(1)})$ we may have EDP convergence to
effective limits $(\sfQ,\cE_0^{(1)},\cR_\eff^{(1)})$ and
$(\sfQ,\cE_0^{(2)},\cR_\eff^{(2)})$, such that the effective equations
\[
\dot c= V^0_{(1)}(c)= \D_\xi \cR^{(1)*}_\eff 
  \big(c,{-}\D\cE_0^{(1)} (c)\big)
\quad \text{and} \quad
\dot c= V^0_{(2)} = \D_\xi \cR^{(2)*}_\eff 
  \big(c,{-}\D\cE_0^{(2)}(c)\big)
\] 
give different dynamics. We refer to 
to \cite[Sec.\,3.3.5.2]{Miel16EGCG} for such a case. \EEE

In \cite[Sec.\,3.3]{LMPR17MOGG} an example of a simple linear reaction
systems (with $I=3$) is considered, where it is shown that the cosh
structure is distinguished by the fact that it is the only one that is
stable under EDP-convergence. It is one of our major results that in
our situation the same stability is true, i.e.\ EDP-convergence yields
a limit gradient structure of cosh-type again.\medskip

We now describe our results more precisely. We mainly work under the
assumption that our system \eqref{eq2} satisfies the DBC
\eqref{eq:I.DBC} for $w^\eps$ and assume that $w^\eps \to w^0 \in
{]0,1[}^I$, i.e.\ all components $w^0_i$ are positive. Then, clearly
$A^F$ satisfies the DBC for $w^0$. As is shown in Section
\ref{se:EpsDependReactNetw}, the fast reactions encoded in $A^F$
separate $\cI=\{1,\ldots, I\}$ into $J<I$ clusters, and we define
a coarse graining operator $M\in \R^{J\times I}$ and a reconstruction
operator $N\in \R^{I\times J}$ satisfying 
\[
MA^F=0\in \R^{J\ti I},\quad A^FN= 0 \in \R^{I\ti J}, \quad \text{and }\ MN=\mathrm{id}_{\R^J}.
\]
The coarse graining operator $M$ satisfies $M_{ji} \in \{0,1\}$
indicating whether the species $i$ belongs to the cluster $j$.  The
limit equation, which is derived in Theorem \ref{thm:ConvergSol}
independently of any EDP-convergence for clarity, then reads
\begin{equation}
  \label{eq:I.LimitEqn}
   M\dot c(t) = M A^S c(t) \qquad \text{and} \qquad A^F c(t)=0.
\end{equation}
Although convergence of solutions of  (\ref{eq2}) is indeed
well-known, we added a short proof, as it shows similarities to the
proof of EDP-convergence in using complementary information to derive
compactness. \EEE 
Using the coarse-grained states $\hat c(t)=Mc(t) \in \hat\sfQ
\subset \R^J$ with probabilities $\hat c_j(t)$ for the cluster $j \in
\cJ$ one obtains the coarse-grained linear reaction systems 
\begin{equation}
 \label{eq:I.CoarseGrain}
  \dot{\hat c}(t) = \hat A \,\hat c(t) \qquad \text{with} \quad 
      \hat A = MA^S N \in \R^{J\ti J}.
\end{equation}
We refer to Section \ref{su:LimitEqn} for a detailed description and an
interpretation of the coarse-grained equation. \EEE

From the solutions $\hat c$ we obtain all solutions
of the limit equation \eqref{eq:I.LimitEqn} via $c(t)=N\hat c(t)$.
In fact, setting $\hat w:=Mw^0\in {]0,1[}^J$ and defining the diagonal
mappings $\DDD_{w^0}= \mathrm{diag} (w^0_i)_{i\in \cI}$ and
$\DDD_{\hat w} = \mathrm{diag} (\hat w_j)_{j\in \cJ}$ the
reconstruction operator $N$ is given via $N= \DDD_{w^0} M^* \DDD_{\hat
  w}^{-1}$. The intrinsic definition of $N$ becomes clear from duality
theory as $\DDD_{w^0}$ can be seen as a duality mapping from relative
densities $\varrho\in (\R^I)^*$ to concentrations $c \in \R^I$. 
\[ 
\begin{tikzcd}[row sep=large, column sep=large]
c\in \R^I \arrow[thick]{r}{\displaystyle \DDD_{w^0}^{-1}} 
\arrow[thick, xshift=-1.3ex,swap]{d}{\displaystyle M} 
  & \varrho \in (\R^I)^* \supset  M^*(\R^J)^*
\\%
\hat c \in \R^J \arrow[thick]{r}{\displaystyle \DDD_{\widehat w}^{-1}} 
 \arrow[red!70!black, xshift=1.3ex, dashed, very thick ]{u}{\displaystyle
   N\hspace*{-1.5em}} 
 & \hat\varrho \in (\R^J)^* \arrow[thick]{u}{\displaystyle M^*}
\end{tikzcd}
\]

In Section \ref{SectionGeneralGS} we discuss general gradient systems
and define different notions of EDP-convergence as in
\cite{DoFrMi19GSWE, MiMoPe18?EFED}, while Section
\ref{SectionGSforLRS} recalls the different possible gradient
structures for linear reaction systems satisfying the DBC
\eqref{eq:I.DBC}. In Section \ref{su:Tilting} we address the important
notion of tilting of Markov processes which means the change of the equilibrium measure $w$ into $w^\eta= \frac1Z\big(
\mathrm e^{-\eta_i} w_i \big){}_{i\in \cI}$. It is another remarkable
feature of the cosh gradient structure that it is invariant under
tilting (see Proposition \ref{pr:Tilt.Inv.GS}). Thus, the strong
notion of EDP-convergence with tilting introduced in
\cite{DoFrMi19GSWE, MiMoPe18?EFED} can only be shown for the
cosh gradient structure. \EEE  

In Section \ref{SectionMainTheorem} we present our main result on the
EDP-convergence with tilting of the cosh-gradient systems
$(\sfQ,\cE^\eps_\mathrm{Bz}, \cR^*_\eps)$ defined via
\eqref{eq:I.cosh}.  While the $\Gamma$-convergence
$\cE^\eps_\mathrm{Bz}\Gammlim \cE^0_\mathrm{Bz}$ follows trivially
from $w^\eps \to w^0$, the $\Gamma$-convergence $\fD_\eps \Gammlim
\fD_0$ in $\rmL^2([0,T], \sfQ)$ is much more delicate. In fact,
Theorem \ref{thm:MoscoCvg.frakD} even provides the Mosco-convergence
of $\fD_\eps \Moscolim \fD_0$, i.e.\ (i) the liminf estimate
$\liminf_{\eps\to 0} \fD_\eps (c^\eps) \geq \fD_0(c^0)$ holds even
under the weak convergence $c^\eps \rightharpoonup c^0$ in
$\rmL^2([0,T];\sfQ)$ and (ii) for each $c^0 \in \rmL^2([0,T];\sfQ)$
there exists a recovery sequence $c^\eps \to c^0$ strongly(!)\ in
$\rmL^2([0,T];\sfQ)$ such that $\limsup_{\eps\to 0} \fD_\eps (c^\eps)
\leq \fD_0(c^0)$.

The main point of the result is the exact characterization
of $\cR_\eff$.  Indeed, we have  
\begin{align*}
 \fD_0(c) = \begin{cases}
                     \displaystyle\int_0^T \!\!\Big(\cR_\eff(c, \dot
                     c) + \cR_\eff^*(c, - \D \cE_\textrm{Bz}^0(c))\Big) \dd t
                     &\text{for } c\in \rmW^{1,1}([0,1];P\sfQ),\\
                     \infty &\text{otherwise in }\rmL^2([0,1];\sfQ),
                    \end{cases}
\end{align*}
where, for $c \in P\sfQ$ the effective dissipation potential
$\cR_\eff$ is given by 
\begin{align*}
 \cR_\eff^*(c,\xi) = \cR_S^*(c,\xi) {+} \chi_{M^*(\R^J)^*}(\xi)
  \quad \text{or equivalently} \quad 
 \cR_\eff(c,v) = \!\! \inf_{z\in\R^I: Mz = Mv} \!\! \cR_S(c, z).
\end{align*}
Here $P=NM$ is the projection mapping general $c \in \sfQ$ into
microscopically equilibrated reactions $c=N\hat c$ with $\hat c=Mc$,
and $\cR_S^*$ is the dual dissipation potential defined as in
\eqref{eq:I.cosh} but using only the slow reactions. Finally, the
characteristic function $\chi_\Xi$ is $0$ for $\xi\in \Xi$ and
$\infty$ else. The condition $\chi_\Xi ({-}\D\cE_\mathrm{Bz}^0 (c)) <
\infty$ is in fact equivalent to $c \in P\sfQ$, see Section
\ref{SectionLimitingGS}). \EEE 

It is easy to see that the degenerate gradient system
$(\sfQ,\cE^0_\text{Bz} ,\cR^*_\eff)$ generates exactly the limit
equation \eqref{eq:I.LimitEqn}. Moreover, 
using the bijective linear mapping $M:P\sfQ \to \hat\sfQ := \bigset{
  \hat c\in [0,1]^J}{ \hat c_1+\cdots + \hat c_J=1} \subset \R^J$ 
with inverse $N:\hat \sfQ \to P \sfQ\subset \R^I$ we can
define the coarse-grained gradient system $(\hat \sfQ, \hat\cE, \hat
\cR)$ for the coarse-grained states $\hat c = Mc$ via 
\begin{align*}
 \hat\cE(\hat c) = \cE_\mathrm{Bz}^0(N\hat c), \quad 
 \hat \cR (\hat c, \hat v) = \cR_\eff(N\hat c, N\hat v), \quad 
  \hat \cR^*(\hat c, \hat \xi) = \cR_\eff^*(N\hat c, M^*\hat\xi).
\end{align*}
The construction and the explicit formula for $\cR^*_\eff$ yield that
$(\hat \sfQ, \hat\cE, \hat \cR)$ is again a cosh gradient structure
and the associated gradient-flow equation is the coarse-grained
equation \eqref{eq:I.CoarseGrain}, see
Proposition~\ref{pr:kappa.hat.kappa}.

This is indeed a rigorous coarse-graining in the sense of
\cite[Sec.\,6.1]{MaaMie18?MCRD}. This paper is intended to be an
easy-to-understand first approach to more general results of 
EDP-convergence. In  \cite{MiPeSt20EDPCNLRS} we will cover
nonlinear reaction systems, for which the coarse-graining
procedure based on Markov operators cannot work and where the structure of
the limiting equations will be more involved because of nonlinear
algebraic constraints. Moreover, reaction-diffusion systems are
discussed in \cite{FreLie19?EDTS, FreMie19?DKRF}. As
already shown in \cite{LMPR17MOGG}, the cosh gradient structure may appear
automatically from quadratic diffusion models, and we \EEE
expect that the cosh gradient structure will also be stable in
the more general situations in \cite{Step20EDPLRDS}, where
also the reaction fluxes are coarse-grained and reconstructed.
\EEE

\section{Fast-slow reaction network}
\label{se:EpsDependReactNetw}

On $\sfQ:=\mathrm{Prob}(\cI ):=\bigset{c\in[0,1]^I}{
\sum_{i\in \cI}c_i=1} \subset X:=\R^I$ we consider 
the Kolmogorov forward equation
or master equation
\begin{align*}
 \dot c = A c \quad \mathrm{with} \quad A \in \R^{I\times I},
\end{align*}
where $A$ is the adjoint of a Markov generator, i.e.
\begin{align*}
  A_{ik}\geq 0 \quad \text{for all } i \neq k \quad
  \text{and} \quad  \forall \,k\in \cI :\quad 0=\sum\nolimits_{i=1}^I A_{ik}.
\end{align*}
Some comments on the notation are in order.  Usually, in the theory of
Markov operators and stochastic processes the state space is the set
of probability measures which is a subset of the dual space of
continuous functions. So it would be more convenient to denote the
space of interest by $X^*$ and not $X$. Certainly, since we are
dealing with finite dimensional spaces, both are isomorphic and the
notation is just a question of manner. In that paper, the master
equation is understood as a \textit{rate equation} of a gradient
system in the sense of Section \ref{SectionGeneralGS} which is an
equation in $X$. Strictly speaking, the operator $A$ is the adjoint of
a Markov generator $\mathcal L$ which generates a semigroup of Markov
operators $\mathrm e^{t\mathcal L }:X^*\to X^*$. 
By definition, a Markov operator $M^*: X^*\rightarrow
Y^*$ on a finite dimensional state space maps positive vectors on
positive vectors and the constant one vector $\one_{X^*}$ to a constant
one vector $\one_{Y^*}$. Its adjoint maps the set of probability vectors
onto the set of probability vectors. 

The linear reactions given by $A$, naturally define a graph or
reaction network, where edges $e_{ik}$ from node $x_i$ to node $x_k$
correspond to the entries $A_{ik}>0$. The graph is directed,
i.e. edges $e_{ik}$ and $e_{ki}$ are different and have an
orientation. We assume that $A$ is irreducible, which means that the
corresponding graph is irreducible, or in other words, that any two
nodes are connected via a directed path. This implies that there
is a unique steady state $w\in\mathrm{Prob}(\cI )$ which is positive,
i.e. $w_j>0$ for all $j\in \cI $, see e.g.\ \cite{Durr10PTE}.

The crucial assumption for our systems is the following symmetry
condition. The Markov process satisfies is called to satisfy the
\emph{detailed-balance condition} (DBC) with respect to its stationary
measure $w>0$, if $A_{ik}w_k=A_{ki}w_i$ for all $i,k\in \cI$. 
Assuming detailed balance, the evolution equation $\dot c = A c$,
which is an equation on $X$, can also be written in another form. Let
us introduce the duality operator 
\[
\DDD_w = \mathrm{diag}(w):\left\{ \begin{array}{ccc} X^*&
\rightarrow &X, \\ \varrho &\mapsto &c=\DDD_{w}\varrho
\end{array} \right. \qquad \text{and} \qquad 
X\ni c \xrightarrow {\DDD_w^{-1}} \varrho \in X^*.
\]
Hence, $\DDD_w$ maps the relative
densities $\varrho$ to the concentrations $c$, i.e. $c_i=\varrho_iw_i$.
The linear master equation can now be written as
\begin{align*}
 \dot c = B \varrho \quad \text{with } B=A\DDD_w\,.
\end{align*}
Because of the DBC, $B=A\DDD_w:X^*\rightarrow X$ is a
symmetric operator on $X$, i.e. $B^* = B$.

For our slow-fast systems, we introduce a scaling parameter
$1/\eps$ for $\eps>0$ and the rates $A_{ik}$ on the right-hand side
decompose into $A = A^\eps = A^S + \frac{1}{\eps}A^F$, where ``$S$''
stands for slow and ``$F$'' for fast reactions. Our equation is
$\eps$-dependent and reads
\begin{align}
 \label{eqceps}
  \dot c^\eps = A^\eps c^\eps = (A^S + \frac 1 \eps A^F)c^\eps.
\end{align}
The aim of the paper is to investigate the system in the limit
$\eps\rightarrow 0$. To do this, some assumptions on the
$\eps$-dependent reaction network are needed.

\subsection{Assumptions on the $\eps$-dependency of the network}
\label{su:AssumpNetw}

Our paper will be restricted to the case where the stationary measure
$w^\eps\in \sfQ$ converges to a positive limit measure
$w^\eps\rightarrow w^0 \in {]0,1[}^I$:
{\def\theequation{2.A}
\begin{subequations}
 \label{eq:cond.ALL}
\begin{align}
 \label{eq:cond.A1}
&\text{\parbox{0.8\textwidth}{
  For all $\eps>0$, the reaction graph defined by $A^\eps$ is
  connected.\\ Moreover, if there is a transition from state $i$ to
  $k$ (i.e.\ $A_{ki}>0$),
  then there is also a transition backwards from $k$ to $i$.}}
\\[0.4em]
& \label{eq:cond.A2}
 \text{\parbox{0.8\textwidth}{
  For all $\eps>0$ there is a unique and positive stationary measure
$w^\eps\in \sfQ$, and the stationary measure converges
$w^\eps\rightarrow w^0$, where $w^0$ is positive.}}
\\[0.4em]
& \label{eq:cond.A.DBC}
\text{\parbox{0.8\textwidth}{
   (DBC):  For all $\eps>0$ the detailed-balance condition with respect to 
   $w^\eps$ holds, i.e.\ $A^\eps_{ik}w^\eps_k = A^\eps_{ki}w^\eps_i$
   for all $i,k\in \cI$.}}
\end{align}
\end{subequations}
}%
These three conditions are not independent of each other, but it
is practical to state them as above. In particular, if
\eqref{eq:cond.A1} and the DBC \eqref{eq:cond.A.DBC} hold, then 
\eqref{eq:cond.A2} follow, which is the content of the following
results. See \cite{Step19?CCSM} and the references therein for
generalizations.

\begin{prop}
\label{prop:DBSimpliesPositivityofw0}
Let the reaction network satisfy \eqref{eq:cond.A1} and
\eqref{eq:cond.A.DBC} and define, for transitions according
\eqref{eq:cond.A1}, the transition quotients 
\begin{align*}
q_{ik}^\eps = \frac{A_{ik}^\eps}{A_{ki}^\eps}= \frac{A_{ik}^S + 
\frac 1 \eps A_{ik}^F}{A_{ki}^S + \frac 1 \eps A_{ki}^F} \:.
\end{align*}
If there is a (universal) bound $q^*<\infty$ such that for all
transitions from $i$ to $k$ and for all $\eps\geq0$ the transition
quotients $q_{ik}^\eps$ satisfy $1/q^* \leq q_{ik}^\eps\leq q^*$, then $w^\eps$
converges and its limit $w^0$ is positive, i.e.\ \eqref{eq:cond.A2}
holds.
\end{prop}
\begin{proof}
  Using the DBC \eqref{eq:cond.A.DBC}, the stationary measure $w^\eps$
  only depends on the transition quotient $q_{ik}^\eps$.  Hence, each
  $\eps\mapsto w^\eps_i \in [0,1]$ is a rational polynomial in $\eps$
  and thus converges to $w^0_i$ with $w^0 \in \sfQ=\mathrm{Prob}(\cI
  )$ with polynomial dependency on $\eps>0$. Moreover,
  $q_{ik}^\eps=1/q_{ki}^\eps$ converges to $q_{ik}^0 \in
  [1/q^*,q^*]$. Since the limit $w^0$ again depends only $q_{ik}^0$,
  we conclude that it is positive.
\end{proof}

We now comment on the relevance of the above assumptions and give two
nontrivial examples. 

\begin{rem}\label{RemarkAssumptionsEpsReactionNetwork}\mbox{} \vspace{-0.6em}
\begin{enumerate}\itemsep-0.3em
\item[(a)] In the chemical literature, our assumption
  \eqref{eq:cond.A1} is often called (weak) \textit{reversibility}. It
  implies already that the stationary measure $w^\eps$ for
  $A^\eps$ is unique and positive.

 \item[(b)] The assumptions in Proposition
   \ref{prop:DBSimpliesPositivityofw0} say that the quotients
   $q_{ij}^\eps$ are bounded even for $\eps\rightarrow 0$ and hence,
   they converge.  In particular, this means that if there is a fast
   reaction $A_{ik}^F\neq 0$ then necessarily also the backward
   reaction is fast, i.e. $A_{ki}^F\neq 0$. So, the graph does not
   change its topology in the limit process $\eps\rightarrow
   0$. Without this assumption the mass $w^\eps_i$ may vanish for some
   species $i$, see Example \ref{ExampleNetworks}(b).  This case is
   more delicate and will be considered in subsequent work.

 \item[(c)] It was observed in \cite{Yong08ICPD, Miel11GSRD} that reaction
   systems of mass-action type have an entropic gradient structure, if
   the DBC holds. For linear reaction
   systems this was independently found in \cite{Maas11GFEF,
     CHLZ12FPEF}. However, our work will not use the quadratic
   gradient structure derived in the latter works, but will rely on the
   cosh-type generalized gradient structure derived in
   \cite{MiPeRe14RGFL, MPPR17NETP}, see Section \ref{SectionGSforLRS}.

 \item[(d)] Assuming \eqref{eq:cond.A1}, \eqref{eq:cond.A.DBC}, and
   additionally that the reaction quotients $q_{ik}^\eps$ scale either
   with $1$ or with $1/\eps$, i.e. $A_{ik}^F\neq 0 \Rightarrow
   A_{ik}^S=0$, then the transition quotients $q_{ik}^\eps$ are
   $\eps$-independent. In particular, the stationary measure $w_\eps$
   as well as the energy $\cE_\eps$ (see Section \ref{su:GenerClass}) are
   independent of $\eps$.
 \end{enumerate}
\end{rem}

\begin{exam}\label{ExampleNetworks} We discuss two cases highlighting
  the relevance of our assumptions. 
 \begin{enumerate}
  \item[(a)] A prototype example is the following, where four states are involved:
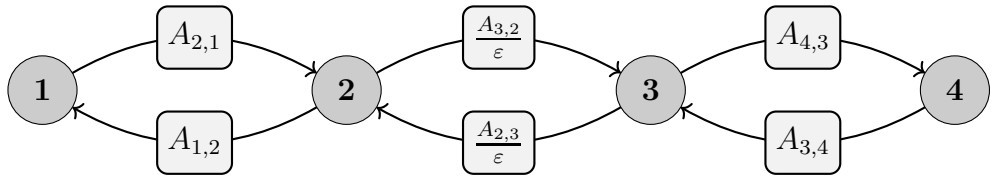
\begin{center}
\begin{tikzpicture}
\SetGraphUnit{4}
  \Vertex{3}
   \WE(3){2}
   \WE(2){1}
   \EA(3){4}
  \Edge[label = $A_{1,2}$](2)(1)
  \Edge[label = $\frac{A_{2,3}}{\eps}$](3)(2)
  \Edge[label = $\frac{A_{3,2}}{\eps}$](2)(3)
  \Edge[label = $A_{2,1}$](1)(2)
  \Edge[label = $A_{4,3}$](3)(4)
  \Edge[label = $A_{3,4}$](4)(3)
\end{tikzpicture}
\end{center}
As in all reaction chains, this example satisfies the DBC \eqref{eq:cond.A.DBC}.
  
We observe that the reaction rates $A_{ik}^\eps$ scale either with 1
or with $1/\eps$ and hence, the reaction ratios as well as the
stationary measure do not depend on $\eps$, see Remark
\ref{RemarkAssumptionsEpsReactionNetwork}(d). Hence, the assumptions
\eqref{eq:cond.ALL} are satisfied. We expect that in the limit $\eps\rightarrow
0$ a local equilibrium between the states 2 and 3 occur, which means
that the system can be described by only three states. 
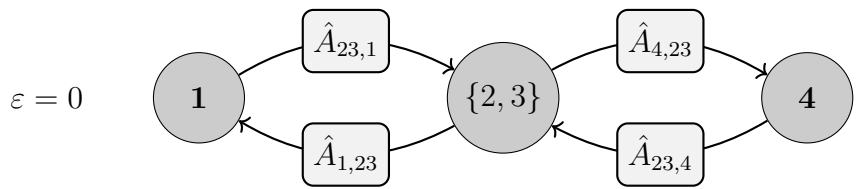
\begin{center}
\begin{tikzpicture}[scale=1.0, transform shape,
state/.style = {circle, draw, fill=black!20, 
                minimum width= 2.9em, font =  \bfseries},
RRR/.style = {rectangle, rounded corners, draw,
              minimum width = 2em, minimum height =2em ,
              fill = black!5, text = black},
 MMM/.style = {}]

\node (a) at (-2, 0) {$\eps=0$};
\node [state]  (N1) at (0,0) {\  1\ \ };
\node [state] (N23) at (4,0) {$\{2,3\}$};
\node [state]  (N4) at (8,0) {\ 4\ \ };

\path[thick,->] (N1) edge[bend left] node[RRR] {$\hat A_{23,1}$} (N23);
\path[thick,->] (N23) edge[bend left] node[RRR] {$\hat A_{4,23}$} (N4);
\path[thick,->] (N4) edge[bend left] node[RRR] {$\hat A_{23,4}$} (N23);
\path[thick,->] (N23) edge[bend left] node[RRR] {$\hat A_{1,23}$} (N1);

\end{tikzpicture}
\end{center}

\item[(b)] In \cite{LMPR17MOGG}, the authors considered the following
  reaction chain:
   \begin{center}
  $\eps>0$\ 
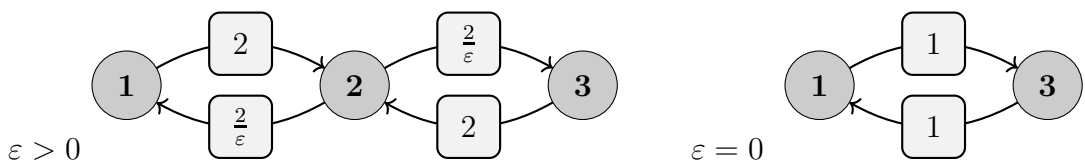
  \begin{tikzpicture}
  \SetGraphUnit{3}
  \Vertex{2}
  \WE(2){1}
  \EA(2){3}
  \Edge[label = 2](1)(2)
  \Edge[label = $\frac 2 \eps$](2)(3)
  \Edge[label = 2](3)(2)
  \Edge[label = $\frac 2 \eps$](2)(1)
 \end{tikzpicture}
    \qquad $\eps=0$ \ 
 \begin{tikzpicture}
  \SetGraphUnit{3}
  \Vertex{3}
  \WE(2){1}
  \Edge[label = $1$](1)(2)
  \Edge[label = $1$](2)(1)
 \end{tikzpicture}
\end{center}
The DBC \eqref{eq:cond.A.DBC} is again satisfied. The stationary measure is
$w_\eps=\frac 1 {2+\eps}(1, \eps, 1)$. The transition quotients are
$q_{12}^\eps=\eps$ and $q_{23}^\eps=\frac 1 \eps$, which converge to
$0$ or $\infty$, respectively. Hence assumption \eqref{eq:cond.A2} is
violated. In fact the limit stationary measure is $w^0 = (\frac 1
2, 0 , \frac 1 2)$, which is no longer strictly positive. In
\cite[Sec.\,3.3]{LMPR17MOGG} the EDP-convergence is performed for
different gradient structures and only the cosh-gradient structure as
defined in Section \ref{suu:Boltz.cosh} turned out to be stable.
 \end{enumerate}

\end{exam}

\subsection{Capturing the states connected by fast reactions}
\label{su:CaptStates}

In the limit species which are connected by fast reactions have to be
treated like one large particle. Let $i_1\sim_F i_2$ denote the
relation if states $i_1$ and $i_2$ are connected via fast reactions.
Assumptions \eqref{eq:cond.A1}, \eqref{eq:cond.A2}, and
\eqref{eq:cond.A.DBC} \EEE guarantee that $\sim_F$ defines an
equivalence relation on $\cI $ and decomposes $\cI $ into different
equivalence classes $\mathcal{J}:=\{\alpha_1, \dots, \alpha_J\}$,
where the index of $\sim_F$, i.e. the number of (different)
equivalence classes, is denoted by $J$. By definition all $\alpha_j$
are non-empty. Obviously, we have $1\leq J \leq I$.  In particular,
$J=I$ means that there are no fast reactions; $J=1$ means that each
two species are connected via at least one reaction path consisting
only of fast reactions. Let
$\phi : \{1, \dots, I\}\rightarrow \{\alpha_1, \dots, \alpha_J\}$ be
the function, which maps a state $i$ to its equivalence class
$\alpha_j$, i.e. $i\mapsto \phi(i)=[i]_{\sim_F}=\alpha_j$. To make
notation simpler, we denote the set of equivalence classes by
$\mathcal{J}=\{1,\dots, J\}$ and further use $j\in \mathcal J$ and
$i\in\mathcal I$.

The function $\phi :\cI \rightarrow \mathcal{J}$ defines a
deterministic Markov operator $M^*: Y^* \rightarrow X^*$, where $Y^*$
is a $J$-dimensional real vector space, by
\begin{align*}
(M^*\hat \varrho)_i := \hat \varrho_{\phi(i)}, ~~~~\hat\varrho \in
Y^*, ~ i \in \cI . 
\end{align*}
Deterministic Markov operator means that its dual $M:X\rightarrow Y$
maps pure concentrations, i.e. unit vectors $e_i$, to pure
concentrations. 

Some facts on deterministic Markov operators are in order. Clearly for
a deterministic Markov operator it holds $M^*(\hat \varrho \cdot \hat
\psi) = M^*\hat \varrho \cdot M^* \hat \psi$ where the multiplication
is meant pointwise. (This, by the way, characterizes all deterministic
Markov operator.) We want to write the multiplicative relation in form
of operators. To do this let us define the multiplication by $\hat
\varrho$ as $\Pi_{\hat \varrho}: Y^* \rightarrow Y^*$, with
$(\Pi_{\hat \varrho} \hat \psi )_j= \hat \varrho_j \cdot \hat
\psi_j$. Hence, we conclude for a deterministic Markov operator that
$M^*\Pi_{\hat \varrho} = \Pi_{M^*\hat \varrho}M^*$. Dualizing this
equation, we get $\Pi_{\hat \varrho}^*M = M\Pi_{M^*\hat
  \varrho}^*$. Note, that the adjoint operator has a simple form:
$\Pi_{\hat \varrho}^* : Y\rightarrow Y$, $\Pi_{\hat \varrho}^*\hat
c=\DDD_{\hat c}\hat\varrho$.  So summarizing
\begin{align}\label{EqDetMOProducts}
\Pi_{\hat \varrho}^*M = M\Pi_{M^*\hat \varrho}^*  \quad \text{and} 
\quad \Pi_{\hat \varrho}^*\hat c=\DDD_{\hat c}\hat\varrho.
\end{align}

In the limit process the species connected by fast reactions are identified. This is done by a linear \textit{coarse-graining-operator}, which is the adjoint of $M^*$, $M: X\rightarrow Y$.
In matrix representation induced by the canonical basis, we have
\begin{align}
\label{eq:Def.cgM}
 M: X\approx\R^I \rightarrow Y \approx \R^J,~~ M_{ji} := \begin{cases}
 1, ~~\mathrm{for~~}i \in \alpha_j,\\
 0, \mathrm{~~otherwise~}.
 \end{cases}
\end{align}
Note that the construction is such that $M$ maps $X\supset\mathrm{Prob}(\cal I)$ onto $Y\supset\mathrm{Prob}(\cal J)$. Since for $\alpha_j$ there is at least one $i$ with $i\in \alpha_j$, the matrix of $M$ has full rank and each column is a unit vector. Moreover, we point out that $M$ and $M^*$ only depend on the reaction network topology and the locations of the fast reactions, the specific reaction rates $A_{ij}$ do not matter (see Example \ref{ExampleOperatorForNetwork} below).

\subsection{Properties of the coarse-graining operator $M$ and
  the reconstruction operator $N$}
\label{su:PropCG+Recov}

Recall the duality map $\DDD_{w^0}$, which is a represented by a diagonal
matrix with entries $w^0>0$, connects the concentrations and the
relative densities, i.e.
\begin{align*}
 \varrho \in X^*\xrightarrow{\DDD_{w^0}} c\in X.
\end{align*}
The subset of $X^*$ which consists of the equilibrated densities
$\varrho_i$ is denoted by $X^*_{\mathrm{eq}}$, i.e.
\begin{align*}
X^*_{\mathrm{eq}} :=\bigset{\varrho\in X^* }
  {\forall\, i_1\sim_F i_2 : \ \varrho_{i_1}=\varrho_{i_2} }.
\end{align*}
\EEE For the limit system, we define the stationary measure (denoted
by $\hat w$) by $\hat w = Mw^0$.  Since $M^*$ is a deterministic
Markov operator, we have the following characterization of the
multiplication operator induced by $\hat w$.

\begin{lem}
\label{le:CommDiagram}
  Let $M^* : Y^*\rightarrow X^*$ be a deterministic Markov operator
  induced by a function $\phi : \{1, \dots, I\}\rightarrow \{1, \dots,
  J\}$ and let $w\in X$. Then $Mw = \hat w$ if and only if $\DDD_{\hat w}
  = M \DDD_{w} M^*$.
\end{lem}
\begin{proof}
Assume that $\DDD_{\hat w} = M \DDD_{w} M^*$ holds. Evaluating both
sides at the constant vector $\one_{Y^*}$, we get $\DDD_{\hat w}
\one_{Y^*} = \hat w$ and $M \DDD_{w} M^* \one_{Y^*} = M \DDD_w \one_{X^*} =
Mw$, since $M^*$ is a Markov operator which maps $\one_{Y^*}\mapsto
\one_{X^*}$. This proves the claim in one direction.

Assume $\hat w = Mw$ we have to show that $\DDD_{M w} = M \DDD_{w}
M^*$.  We use statement (\ref{EqDetMOProducts}) for deterministic
Markov operators and find $ \DDD_{Mw} \hat \varrho = \Pi_{\hat
  \varrho}^* Mw = M\Pi_{M^*\hat \varrho}^*w = M \DDD_w M^* \hat
\varrho$.
\end{proof}
If $M^*$ is not a deterministic Markov operator but a
general one, then the above relation will not hold.

We assumed that all equivalence classes $\alpha_j$ are non-empty and
hence, each row of our coarse-graining operator $M$ defined in
\eqref{eq:Def.cgM} \EEE has at least one entry ``$1$''. In particular, this
implies that $\hat w$ is strictly positive and hence, $\DDD_{\hat w}$ is
invertible.  In particular, we proved that the following diagram
commutes:
\[ \begin{tikzcd} c\in X \arrow{r}{\DDD_{w^0}^{-1}} \arrow[swap]{d}{M} &
  \varrho \in X^* \supset X^*_{\mathrm{eq}} =\bigset{\varrho\in X^* }
  {\forall\, i_1\sim_F i_2 : \ \varrho_{i_1}=\varrho_{i_2} } \\%
  \hat c \in Y \arrow{r}{\DDD_{\hat w}^{-1}} & Y^*\arrow{u}{M^*}
\end{tikzcd}
\]

The crucial object is the following operator $N:Y\rightarrow X$, which
"inverts" the coarse-graining operator $M:X\rightarrow Y$, by
mapping coarse-grained concentrations $\hat c\in Y$ to concentrations
$c\in X$ (see also \cite{Step13IMOM}, where the operator is introduced
for its connection to the direction of time). We call $N$ a
\emph{reconstruction operator} as it reconstructs the full information on the
density $c\in X$ from the coarse-grained vector $\hat c\in Y$
assuming, of course, microscopic equilibrium.  More precisely,
 $N$ is defined via 
\begin{align}
\label{eq:def.N.N*}
N := \DDD_{w^0} M^* \DDD_{\hat w}^{-1}:Y\rightarrow X \quad
\text{such that} \quad  N^* = \DDD_{\hat w}^{-1}M\DDD_{w^0}:
X^*\rightarrow Y^*.
\end{align}
While the coarse-graining operator $M$ simply merges the masses
within the corresponding equivalence classes, the reconstruction
operator $N$ redistributions the masses in each equivalence class
proportional to the equilibrium measure. As a result $P=NM$ will be a
projection from general states to states in local equilibrium.

These and other important properties of the operator $M$ and $N$ and
their adjoints $M^*$ and $N^*$ \EEE are summarized in the next
proposition, which is independent of the generators
$A^\eps =A^S+\frac1\eps A^F$.

\begin{prop}
\label{pr:OperatorN}
Let $M^*:Y^*\rightarrow X^*$ be a deterministic Markov operator as
in Lemma \ref{le:CommDiagram} with adjoint $M:X\rightarrow Y$ and 
let $\hat  w :=Mw^0$ for some $w^0 \in {]0,1[}^I\subset \sfQ$. 
Moreover, $N$ and $N^*$ be defined as in \eqref{eq:def.N.N*}, then 
the following holds:
\begin{enumerate}
\item $N^*$ is a Markov operator.
\item $MN=\mathrm{id}_Y$ or $N^*M^* = \mathrm{id}_{Y^*}$, i.e. $N^*$
  is a left-inverse of the Markov operator $M^*$. 
\item $P:=NM$ \EEE is a projection on $X$, which leaves the range of
  $\DDD_{w^0}M^*: Y^*\rightarrow X$ invariant. The adjoint $M^*N^*$ is a
  projection as well, which leaves the range of $M^*$ invariant. 
\item $N\hat w = w^0$, i.e. $N$ inverts w.r.t. the stationary measure.
\item The operator $P^*= M^*N^*$ is a Markov operator on $X^*$ and
  its adjoint $P= NM$ has the stationary measure $w^0$. Moreover,
  $P^*$ satisfies detailed balance w.r.t. $w^0$.
\end{enumerate}
\end{prop}
\begin{proof}
 Clearly, $N^*$ is non-negative and $N^*\one_{X^*} = \DDD_{\hat  w}^{-1}
 M \DDD_{w^0}\one_{X^*} = \DDD_{\hat w}^{-1}Mw^0 =
\one_{Y^*}$ holds. This proves the first statement. 

Lemma \ref{le:CommDiagram} implies that $MN=\mathrm{id}_Y$
and that $NM$ is a projection on $X$, which leaves the range of
$\DDD_{w^0}M^*: Y^*\rightarrow X$ invariant. The fourth claim is also
trivial. It is also not hard to see that $P^*$ is a Markov operator
and that its adjoint has the stationary measure $w^0$. Moreover,
detailed balance holds:
 \begin{align*}
\DDD_{w^0}P^* = \DDD_{w^0}M^* N^* = \DDD_{w^0}M^*\DDD_{\hat
  w}^{-1}M\DDD_{w^0} = NM\DDD_{w^0} = P\DDD_{w^0}. 
\end{align*}
This proves the result. 
\end{proof}

The following example shows how the operators look like in a specific
case. 

\begin{exam}\label{ExampleOperatorForNetwork}
  For the reaction network in Example \ref{ExampleNetworks}(a) we have
  $I=4$ with only one fast reaction $2\sim_F 3$, hence $J=3$. Using
  the numbering $\alpha_1=\{1\}$, $\alpha_2=\{2,3\}$, and
  $\alpha_3=\{4\}$ and the stationary measures
  $w=(w_1,w_2,w_3,w_4)^\top \in X$ and $\hat w = (w_1,w_2{+}w_3,
  w_4)^\top \in Y$, respectively, we find
\begin{align*}
  M=\begin{pmatrix}
      1 & 0 & 0 & 0\\
      0 & 1 & 1 & 0\\
      0 & 0 & 0 & 1
     \end{pmatrix}, 
\ \ 
 N= \begin{pmatrix}
       1 & 0 & 0 \\
       0 & \frac{w_2}{w_2+w_3} & 0 \\
       0 & \frac{w_3}{w_2+w_3} & 0 \\
       0 & 0 & 1
     \end{pmatrix},
\text{ and } 
P=NM=\begin{pmatrix}
       1 & 0 &  0 & 0 \\
       0 & \frac{w_2}{w_2+w_3} &  \frac{w_2}{w_2+w_3} &0 \\
       0 & \frac{w_3}{w_2+w_3} & \frac{w_3}{w_2+w_3} & 0 \\
       0 & 0 & 0& 1
     \end{pmatrix}.
  \end{align*}
\end{exam}

\subsection{The limit equation and the coarse-grained equation}
\label{su:LimitEqn}

As a direct consequence of Proposition \ref{pr:OperatorN} we obtain a
decomposition of the state space $X\approx \R^I$ into the
microscopically equilibrated states 
\[
c = Pc \in \sfQ_\mathrm{eq}:=P\sfQ \subset 
  X_\mathrm{eq}:=PX=\bigset{ c \in X}{A^Fc =0},
\]
which are measures having constant density with respect to $w^0$,
and the \EEE component $(I{-}P)c \in X_\mathrm{fast}:=(I{-}P)X$ that
disappears exponentially on the time scale of the fast reactions.  We
emphasize that the following result does not use the DBC
\eqref{eq:cond.A.DBC}.

\begin{prop}
\label{pr:DecompX}
Under the assumptions \eqref{eq:cond.A1}--\eqref{eq:cond.A2} we have
\begin{subequations}
  \label{eq:DecompoX}
 \begin{align}
    \label{eq:DecompoX.AA}
   & PA^F=A^FP=0 \in \R^{I\ti I},\quad MA^F = 0 \in \R^{J\ti I}, 
  \quad A^FN=0 \in \R^{I\ti J}, 
  \\
   & \label{eq:DecompoX.a} 
     X= X_\mathrm{eq} \oplus X_\mathrm{fast} \quad \text{with }\\
    \label{eq:DecompoX.b}
   & X_\mathrm{eq} = \mathrm{kernel}(A^F) =
             \mathrm{range}(P)=\mathrm{range}(N) \text{ and}\\
  \label{eq:DecompoX.c}
   & X_\mathrm{fast} = \mathrm{range}(A^F) =
             \mathrm{kernel}(P)=\mathrm{kernel}(M).
 \end{align}
\end{subequations}
Here, $X_\mathrm{fast}$ depends on $M$ only, i.e.\ only on the
reaction graph of $A^F$, whereas $X_\mathrm{eq} $ depends on $A^S$ and
$A^F$ through $w^0$. 
\end{prop}
\begin{proof} By construction of $M$ from the reaction network induced
by $A^F$ we immediately obtain $\mathrm{range}(A^F) =
\mathrm{kernel}(M)$. Indeed, the entries of $M$ are all $0$ or $1$,
where the $j$th row contains only the entry $1$ exactly for $i\in
\alpha(j)$. Thus, these $1$s correspond to the mass conservation in
the corresponding equivalence class $\alpha(j)\subset \{1,\ldots,
I\}$, and $MA^F=0$ follows, which implies $ \mathrm{range}(A^F) 
\subset \mathrm{kernel}(M)$. Dimension counting gives the desired
equality.
 
Using the injectivity of $N$ and $P=NM$ we have shown
\eqref{eq:DecompoX.c}.

To establish the relation for $X_\mathrm{eq} $ it suffices to show
$\mathrm{kernel}(A^F)= \mathrm{range}(N)$, since the surjectivity of
$M$ and $P=NM$ gives $ \mathrm{range}(N)= \mathrm{range}(P)$.

Using the dimension counting it is even sufficient to show
$A^FN=0$. Firstly, we use
$0=A^\eps w^\eps =(A^S+\frac 1 \eps A^F)w^\eps$, which gives
$A^Fw^\eps \to 0$, and hence $A^Fw^0=0$. Moreover, we observe that the $j$th column of
$N=\DDD_{w^0} M^* \DDD_{\hat w}$ contains the unique equilibrium
measure associated with the equivalence class $\alpha(j)\subset
\{1,\ldots, I\}$, which implies that $A^FN=0$.\EEE%
\end{proof}

Based on the above result we can formally pass to the limit in our
linear reaction system $\dot c^\eps = (A^S{+}\frac1\eps A^F)c^\eps$. Multiplying
the equation from the left by $M$ we can use $MA^F=0$ and see that
the term of order $\frac1\eps$ disappears. Moreover, it is
expected that the fast reactions equilibrate, so in the limit $\eps\to
0$ we expect the microscopic equilibrium condition $A^F c^\eps \to
0$. Hence, we expect that $c^\eps:[0,T]\to \sfQ$ converges to a function
$c^0:[0,T]\to \sfQ$ which solves the limit equation
\begin{equation}
  \label{eq:LimitEqn}
  M \dot c(t) = MA^S c(t) \quad \text{ and } \quad A^F c(t)=0. 
\end{equation}
Before giving a proof for the convergence $c^\eps \to c$
 we want state that this
system has a unique solution for each initial condition $c(0)$ that is
compatible, i.e.\ $A^F c(0)=0$ and that this solution is characterized
by solving the so-called \emph{coarse-grained equation}. 

\begin{thm}[Coarse-grained equation]
\label{th:CoarseGrEq} For each $c_0\in \sfQ$ with $A^Fc_0=0$ there is a
unique continuous solution $c:[0,T]\to \sfQ$ of \eqref{eq:LimitEqn} with $c(0)=c_0$.  
This solution is obtained by solving the coarse-grained ODE
\begin{equation}
  \label{eq:CoarseGraEqn}
  \dot{\hat c} = MA^S N \, \hat c, \quad \hat c(0)=Mc_0
\end{equation}
and setting $c(t)=N \hat c(t)$.  Moreover, the stationary solution is
$\hat w=Mw^0$. 
\end{thm} 
\begin{proof} On the one hand, 
by \eqref{eq:DecompoX.b} we know that $A^F c=0$ is equivalent to $
c=Pc=NMc$.  Thus, for any solution $c$ of \eqref{eq:LimitEqn}
the coarse-grained state $\hat c = M c$ satisfies the coarse-grained
equation \eqref{eq:CoarseGraEqn}.

On the other hand, \eqref{eq:CoarseGraEqn} is a linear ODE in $\hat
\sfQ\subset Y$ which has a unique solution satisfying $\hat c(t)\in \hat
\sfQ$. This proves the first result.

To see that $\hat w=Mw^0$ is a stationary measure, we use $A^Fw^0=0$ and
\eqref{eq:DecompoX.a}  implies $Pw^0=w^0$. 
On the other hand using $MA^F=0$ we can pass to the limit in $0
=M0=MA^\eps w^\eps= MA^S w^\eps$ to obtain $MA^Sw^0=0$. Combing the
two results we find 
\[
\hat A \hat w = M A^S N(Mw^0)=MA^S\,Pw^0= MA^S w^0=0,
\]
which is the desired result.  
\end{proof}

We emphasize that the coarse-grained equation \eqref{eq:CoarseGraEqn}
is again a linear reaction system, describing the master equation for
a Markov process on $\cJ=\{1,\ldots,J\}$. The effective operator
$\hat A:=M A^S N$ can be interpreted in the following way: $N$ divides
the coarse-grained states into microscopically equilibrated states,
$A^S$ is the part of the slow reactions, and $M$ collects the states
according to their equivalence classes $\alpha(j)$. 

Using $M_{ji}=\delta_{j\phi(i)}$ and $N_{ij}=\frac{w^0_i}{\hat w_j}
\delta_{j\phi(i)} $ the coefficients of
the generator $\hat A = MA^SN$ are easily obtained by a suitable
average, namely
\begin{equation}
  \label{eq:hatA.j1j2}
  \hat A_{j_1j_2} = \sum_{i_1\in \alpha_{j_1}}   \sum_{i_2 \in
    \alpha_{j_2}}  A^S_{i_1 i_2} \frac{w^0_{i_2}}{\hat w_{j_2}}.  
\end{equation}

\subsection{Convergence of solutions on the level of the ODE}
\label{su:ConvergSol}

Finally, for mathematical completeness, we provide a simple and short
convergence proof. It can also be obtained as a special case of the
result in \cite{Both03ILRC}. Of course, the convergence of solutions
is also a byproduct of the EDP-convergence given below, see Lemma
\ref{le:EDPimpliesCvg}. The latter result, which is the main goal of
this work, provides convergence of the gradient structures, which is a
significantly stronger concept, because the coarse-grained equation
\eqref{eq:CoarseGraEqn} has many different gradient structures, while
the EDP-limit is unique.

\begin{thm}[Convergence of $c^\eps$ to $c^0$]
\label{thm:ConvergSol}
Assume \eqref{eq:cond.ALL} and 
consider solutions $c^\eps:[0,T]\to \sfQ$ of \eqref{eq2} such that
$Mc^\eps(0)\to \hat c_0$. Then, we have the convergences 
\[
  Mc^\eps \to  Mc^0 \text{ in } \rmC^{0}([0,T];X)
  \qquad \text{and}  \qquad 
  c^\eps \to c^0 \text{ in }\rmL^2([0,T];X),
\]
where $c^0$ is the unique solution of \eqref{eq:LimitEqn} with
$c^0(0)=N\hat c_0$.  
\end{thm}
\begin{proof} 
  \emph{Step 1: Weak compactness.}  
We first observe that $c^\eps:[0,T] \to \sfQ \subset [0,1]^I$ provides a
trivial a priori bound for $c^\eps$ in
$\rmL^\infty([0,T];\R^I)$. Hence, we may choose a subsequence (not
relabeled) such that $c^\eps \to c^0$ weakly in $\rmL^2([0,T];\R^I)$. 

\emph{Step 2: Compactness of coarse-grained concentrations.}  
With Step 1 we see that $\hat a{}^\eps := M
c^\eps$ is bounded in $\rmC^\mathrm{Lip}([0,T];\R^I)$, because of
$\dot{\hat a}{}^\eps = M\dot c^\eps = MA^S c^\eps$. Thus, there is a
subsequence (not relabeled) such that $\hat a{}^\eps \to \hat a{}^0$
in $\rmC^0([0,T];\R^J)$ and $\hat a{}^0(0)=\hat c_0$. Moreover, with
Step 1 we have $\hat a{}^0 = Mc^0$.

\emph{Step 3: Generation of microscopic equilibrium.} We take the dot
product of the ODE with the vector of relative densities
$c^\eps/w^\eps:=(c^\eps_i/w^\eps_i)_{i=1,..,I}$. Defining the
quadratic form $\cB_\eps(c)=\sum_{i=1}^I \frac {c_i^2}{2 w_i^\eps}$ we
obtain
\begin{align}
 \label{eq:Quadratic} 
  \frac{\d}{ \dd t} \cB_\eps(c^\eps)& = \dot{c}^\eps\cdot
  \frac{c^\eps}{w^\eps} = \big( A^\eps c^\eps) \cdot
  \frac{c^\eps}{w^\eps} 
  =  \frac1\eps \big( B^\eps c^\eps) \cdot c^\eps \quad \text{with } 
  \eps \DDD_{w^\eps}^{-1}A^\eps =:B^\eps = (B^\eps)^* \geq 0.
\end{align}
The latter relations follow from the DBC
\eqref{eq:cond.A.DBC}. Defining the quadratic functional $ \mathfrak Q_\eps(c):=
\int_0^T B^\eps c(t)\cdot c(t) \dd t $ and integrating
\eqref{eq:Quadratic} over $[0,T]$ gives
\[
\mathfrak Q_\eps(c^\eps) = \eps \cB(c^\eps(0)) - \eps \cB(c^\eps(T))
\leq C_1 \eps.
\]
Moreover, using $|w^\eps {-} w^0|\leq C_2\eps$ we find $
|\mathfrak Q_\eps(c) {-} \mathfrak Q_0(c)| \leq C_3 \eps$. 
Hence $\mathfrak Q_0(c^\eps) \leq \mathfrak Q_\eps (c^\eps) + C_3\eps
\leq C_1\eps + C_3 \eps$. 
Using the convexity of $\mathfrak Q_0$ the weak limit $c^0$ of
$c^\eps$ satisfies 
\[
0 \leq \mathfrak Q_0(c^0) \leq \liminf_{\eps \to 0} \mathfrak Q_0(
c^\eps) \leq \liminf_{\eps \to 0} (C_1{+}C_3)\eps = 0.
\]
Since $B^0 = \DDD_{w^0}^{-1} A^F $ is symmetric and positive semidefinite
we conclude $A^F c^0(t)=0$ a.e.\ in $[0,T]$. 
More precisely, by \eqref{eq:DecompoX.c} $c \mapsto \big(B^0c \cdot 
c \big)^{1/2}$ defines a norm on $X_\mathrm{fast}$ that is equivalent
to $c \mapsto |(I{-}P)c|$. Thus, we conclude $(I{-}P)c^\eps \to
(I{-}P)c^0$.  Moreover, Step 2 gives $Pc^\eps = NMc^\eps= N \hat
a{}^\eps \to NM c^0=Pc^0$ such that $c^\eps \to c^0$ in
$\rmL^2([0,T;\R^I)$ follows. 

\emph{Step 4. Limit passage in the ODE.} To see that $c^0$ satisfies
the limit equation \eqref{eq:LimitEqn} we pass to the limit in 
\[
Mc^\eps(t)=Mc^\eps(0)+ \int_0^t M A^S c^\eps(s) \dd s,
\] 
where the left-hand side converges by Step 2 and the right-hand side
by the assumption on the initial condition and by Step 3 and
Lebesgue's' dominated convergence theorem. Thus, $Mc^0(t)=Mc^0(0)+
\int_0^t MA^S c^0(s) \dd s$, and with $A^Fc^0=0$ from Step 3 the
desired limit equation \eqref{eq:LimitEqn} is established. 

As we already know that the solution of \eqref{eq:LimitEqn} is unique,
we conclude convergence of the whole family $(c^\eps)_{\eps>0}$,
instead of a subsequence only.  
\end{proof}

In the above proof the DBC \eqref{eq:cond.A.DBC} is not really
necessary, but it simplified our proof considerably.

\section{Generalized gradient structures}
\label{SectionGeneralGS}

This small section provides the general notions of gradient systems,
gradient-flow equations, the energy-dissipation principle (EDP), and
the three notions of EDP convergence. We follow the survey article
\cite{Miel16EGCG} and the more recent works
\cite{DoFrMi19GSWE,MiMoPe18?EFED}.   

\subsection{Gradient systems and the Energy-Dissipation Principle}
\label{su:GeneralGS}

A triple $(\sfQ, \cE,\cR)$ is called a \emph{gradient system} if
\begin{itemize}
 \item $\sfQ$ is a closed convex subset of a Banach space $X$,
 \item $\cE:\sfQ\rightarrow \R_\infty:=\R\cup\{\infty\}$ is a
   differentiable functional (e.g.\ free energy, negative entropy)
 \item $\cR: \sfQ\times X \rightarrow \Rinfty$ is a dissipation
   potential, i.e.\ for all $u\in \sfQ$ the functional $\cR(u,\cdot ): X
   \rightarrow \Rinfty$ is lower semicontinuous (lsc), nonnegative,
   convex and satisfies $\cR(u,0)=0$. 
\end{itemize}
(More general, $\sfQ$ can be a manifold, then $\cR$ is defined on the
tangent bundle $\mathrm T \sfQ$, but this generalization is not needed in
this work.)
A gradient system $(\sfQ, \cE,\cR)$ is called \textit{classical} if
$\cR(u,\cdot)$ is quadratic, i.e.\ if there are symmetric and positive
definite operators $\mathbb{G}(u):X\to X^*$ such that $\cR(u, v) = \frac 1 2\la
\mathbb{G}(u) v, v \ra$. But often $\cR(u,\cdot)$ is not quadratic
(e.g. for rate-independent processes such as elastoplasticity), see
\cite{Miel16EGCG} and reference therein. We define the dual
dissipation potential $\cR ^*$ using the Legendre transform via
\[
\cR^*(u,\xi) = (\cR(u,\cdot))^*(\xi):= \sup \bigset{\la \xi, v\ra -
\cR(u,v) } {v\in X } .
\] 
The gradient system is uniquely described by  $(Q,\cE, \cR)$
or, equivalently by $(Q,\cE, \cR^*)$ \EEE and, in particular, in this paper
we prefer \EEE the second representation.

The evolution of the states $u(t)$ in a gradient system are given in
terms of the so-called \emph{gradient-flow equation} that is given in
terms of $\cE$ and $\cR$ and can be formulated in three equivalent
ways:
\begin{equation}
  \label{eq:GradFlowEqn}
  \begin{aligned}
 \text{(I) }& \text{\em force balance in $X^*$.} \quad
   0\in\partial_{\dot u}\cR(u,\dot u) + \D\cE(u) \in X^*, 
 \\
 \text{(II) }& \text{\em power balance in $\R$.} \quad \cR(u,\dot u) +
   \cR^*(u, -\D\cE(u)) = -\la \D\cE(u), \dot u \ra, 
 \\
 \text{(III) } & \text{\em rate equation in $X$.} \quad \dot u
   \in \partial_{\xi}\cR^*(u, -\D\cE(u)) \in X, 
 \end{aligned}
\end{equation}
where $\partial$ is the set-valued partial subdifferential with
respect to the second variable.  

In general, we cannot expect that the solution of the
gradient-flow equation fill the whole state space. Clearly, along
solutions we want to have $\cE(u(t))<\infty$ for $t>0$. Moreover,
relation (III) asks that $-\D\cE(u(t))$ lies in the domain of
$\partial_\xi\cR^*(u(t),\cdot)$ for a.a.\ $t\in [0,T]$. Thus, we set 
\begin{equation}
  \label{eq:Dom.GS}
  \mathrm{Dom}(\sfQ,\cE,\cR):=\bigset{u \in \sfQ}{ \D\cE(u)\text{
      exists},\  \partial_{\xi}\cR^*(u, -\D\cE(u))\text{ is nonempty } }.
\end{equation}
Typically, one expects that solutions exist for all initial conditions
in the closure of $\mathrm{Dom}(\sfQ,\cE,\cR)$. 

These three formulations are the same due to the so-called
\textit{Fenchel equivalences} (cf.\ \cite{Fenc49CCF}): Let $Z$ be a
reflexive Banach space and $\Psi: Z\rightarrow \Rinfty$ be a proper,
convex and lsc, then for every all pairs $(v,\xi) \in Z \ti Z^*$ the
following holds:
\[
\text{(i) } \xi\in\partial \Psi(v) \quad \Longleftrightarrow \quad 
\text{(ii) } \Psi(v) + \Psi^*(\xi) = \la \xi, v \ra \quad \Longleftrightarrow \quad
\text{(iii) } v\in\partial\Psi^*(\xi).
\]
We emphasize that (ii) and (II) should be seen as scalar optimality
conditions, because the definition of the Legendre transform easily
gives the Young-Fenchel inequality, namely $\Psi(v)+\Psi^*(\xi)\geq \la \xi,v\ra$ for
all $(v,\xi)\in Z\ti Z^*$. 

Integrating the power balance (II) in \eqref{eq:GradFlowEqn} over
$[0,T]$ along a solution $u:[0,T]\to \sfQ$ and using the chain rule $\la
\D\cE(u(t)), \dot u(t) \ra= \frac{\d}{\d t} \cE(u(t))$ we find the
\textit{Energy-Dissipation Balance} (EDB):
\begin{align}
\label{eq:EDB.gen}
\cE(u(T)) + \int_0^T \! \Big(\cR(u(t),\dot u(t)) +\cR^*(u(t),
-\D\cE(u(t))) \Big) \dd t = \cE(u(0)). 
\end{align}
The following \emph{Energy-Dissipation Principle} (EDP) states that
solving \eqref{eq:EDB.gen} is equivalent to solving the gradient-flow
equation \eqref{eq:GradFlowEqn}.

\begin{thm}[Energy-dissipation principle, see e.g.\
  {\cite[Th.\,3.2]{Miel16EGCG}}]
\label{thm:EDP}
Assume that $\sfQ$ is a closed convex subset of $X=\R^I$, that $\cE\in
\rmC^1(\sfQ,\R)$, and that the dissipation potential $\cR(u,\cdot)$ is
superlinear uniformly in $u\in \sfQ$. Then, a function $u \in \mathrm
W^{1,1}([0,T];\sfQ)$ is a solution of the gradient-flow equation
\eqref{eq:GradFlowEqn} if and only if $u$ solves the
energy-dissipation balance \eqref{eq:EDB.gen}.
\end{thm}

Again, the EDB is an optimality condition, because integrating the
Young-Fenchel inequality for arbitrary $\widetilde u \in \mathrm
W^{1,1}([0,T];\sfQ)$ and using the chain rule we obtain the estimate 
\begin{align}
\label{eq:EDEst.gen}
\cE(\widetilde u(T)) + \int_0^T \! \Big(\cR(\widetilde
u(t),\dot{\widetilde u}(t)) +\cR^*(\widetilde u(t),
-\D\cE(\widetilde u(t))) \Big) \dd t \geq  \cE(\widetilde u(0)). 
\end{align}

The above considerations show that an important quantity associated
with a gradient system $(\sfQ,\cE,\cR)$ is given by the \emph{dissipation
  functional} 
\begin{align*}
 \fD(u):=\int_0^T \! \Big(\cR(u(t),\dot u(t)) +\cR^*\big(u(t),
 -\D\cE(u(t)) \big) \Big)  \dd t, 
\end{align*}
which is defined for all curves $u \in \mathrm
W^{1,1}([0,T];\sfQ)$.

\subsection{General gradient systems and EDP-convergence}
\label{su:EDPcvg.General}

In the following, we consider a family of gradient systems $(X,
\cE_\eps, \cR_\eps)$ and define a notion of convergence on the level
of gradient systems which uniquely defines the limit or effective
system $(\sfQ,\cE_0,\cR_\eff)$. 
Our notion relies on the the energy-dissipation principle from above
and the so-called sequential $\Gamma$-convergence for
functionals, which is defined as follows. 


\begin{defi}[$\Gamma$-convergence, see e.g.\ \cite{Atto84VCFO}]
 For functionals $(I_\eps)_{\eps>0}$ on a Banach space $Z$ we
 say $I_\eps$ \emph{(strongly) $\Gamma$-converges to} $I$, and write
 $I_\eps\Gammlim I$, if the following two conditions hold: 
 \begin{enumerate}
 \item \emph{Liminf estimate.}\\ \indent if $u_\eps \to u$ in $Z$,
   then \quad $I(u) \ \leq \ \liminf_{\eps \to 0} I_\eps(u_\eps)$,
 \item \emph{Existence of recovery sequences.}\\ 
  \indent for all $\tilde u \in Z$ there
   exists $(\tilde u_\eps)_{\eps>0} $ such that \quad $\tilde u_\eps \to \tilde
   u$ \ and \ $\lim_{\eps\rightarrow 0} I_\eps(\tilde u_\eps) = I(\tilde u)$.
 \end{enumerate}
 If the same conditions hold when the strong convergences ``$\to$''
 are replaced by weak convergences ``$\rightharpoonup$'', we say that
 $I_\eps$ \emph{weakly $\Gamma$-converges to} $I$ and write
 $I_\eps\wGammlim I$. If $I_\eps\wGammlim I$ and $I_\eps\Gammlim I$
 holds, we say that $I_\eps$ \emph{Mosco converges to} $I$ and write
 $I_\eps \Moscolim I$.
\end{defi}
\noindent Clearly, for finite-dimensional Banach spaces $Z$ the convergences 
$ \Gammlim $, $ \wGammlim$, and $\Moscolim$ coincide.

The energy dissipation principle allows us to formulate the
gradient-flow equation in terms of the two functionals $\cE_\eps$ and
$\fD_\eps$. However, to explore the full structure of gradient systems
it is useful to embed the given gradient system into a family of
tilted gradient systems $(\sfQ,\cE^\eta,\cR)$, where the \emph{tilted
  energies} $\cE^\eta$ are given by  
\begin{equation}
  \label{eq:cE.tilt}
  \cE^\eta(u)= \cE(u)- \ell^\eta(u) \quad \text{with }
  \ell^\eta(u):=\la \eta, u\ra
\end{equation}
with an arbitrary tilt $\eta \in X^*$. Moreover, introducing the
\emph{tilted dissipation functional} 
\begin{equation}
  \label{eq:mfDeps}
  \fD^\eta_\eps(u):=\int_0^T \! \Big(\cR_\eps(u,\dot u) + 
             \cR_\eps^*(u, \eta{-}\D\cE_\eps(u)) \Big)  \dd t,
\end{equation}
we can now define three versions of EDP-convergence for a
family $\big( (\sfQ,\cE_\eps,\cR_\eps)\big)_{\eps>0}$ as follows. 

\begin{defi}[EDP-convergence \cite{DoFrMi19GSWE,MiMoPe18?EFED}]
\label{def:EDPcvg}
  Let $\sfQ$ be a closed convex subset of a Banach space $X$ and let
  $\cE_\eps$ be Gateaux differentiable.\smallskip
 
\noindent (A) We say that the gradient systems $(\sfQ, \cE_\eps, \cR_\eps )_{\eps>0}$
  \emph{converges in the simple EDP sense to} $(\sfQ, \cE_0, \cR_\eff)$,
  and write $(\sfQ, \cE_\eps, \cR_\eps ) \EDPto (\sfQ, \cE_0, \cR_\eff )$,
  if the following conditions hold:
\begin{enumerate}\itemsep-0.2em
  \item[(i)\ ] $\cE_\eps\wGammlim \cE_0$ on $\sfQ\subset X$,  and
  \item[(ii)\ ] $\fD_\eps\wGammlim \fD_0$ on $\rmL^2([0,T];\sfQ)$ with 
    $\fD_0(u)=\int_0^T\!\!\big(\cR_\eff(u,\dot u) +
    \cR_\eff^*(u,-\D\cE_0(u))\big) \dd t$. \smallskip
\end{enumerate}
(B) We say that $(\sfQ, \cE_\eps, \cR_\eps )$
\emph{EDP-converges with tilting to} $(\sfQ,
 \cE_0, \cR_\eff)$, if for all $\eta\in X^*$ we have 
$(\sfQ, \cE_\eps{-}\ell_\eta, \cR_\eps ) \EDPto (\sfQ, \cE_0{-}\ell_\eta,
\cR_\eff )$.\medskip
 
\noindent (C) We say that $(\sfQ, \cE_\eps, \cR_\eps )$
\emph{contact EDP-converges with tilting to} $(\sfQ,
 \cE_0, \cR_\eff)$, if (i) holds and for all $\eta\in X^*$ we
 have $\fD_\eps^\eta\wGammlim
 \fD_0^\eta$ with $\fD_0^\eta(u)=\int_0^T \mathcal M(u(t),\dot
 u(t),\eta{-}\D\cE_0(u(t))\big) \dd t$, where $\mathcal M$ satisfies
 the contact conditions 
\begin{align*}
 \text{(c1) }\ & \mathcal M(u,v,\xi)\geq \la \xi,v\ra
    \text{ for all }(v,\xi) \in X\ti X^*,\\ 
 \text{(c2) }\ & \mathcal M(u,v,\xi)=\la \xi,v\rangle \
 \Longleftrightarrow \  \cR_\eff(u,v){+}\cR_\eff^*(u,\xi)=\la\xi,v\ra. 
\end{align*} 
\end{defi}
Clearly, `tilted EDP-convergence' is a stronger notion than
`contact EDP-convergence' since the contact potential $\mathcal M$ is
explicitly given in $\cR{+}\cR^*$ form.  We refer to
\cite{DoFrMi19GSWE,MiMoPe18?EFED} for a general discussions of
EDP-convergence and remark that `contact EDP-convergence with tilting'
was called `relaxed EDP-convergence' in \cite{DoFrMi19GSWE}. We
emphasize that there are cases where we have the $\Gamma$ (or even
Mosco) convergence $\cR_\eps \to \cR_0$, but EDP-convergence yields
$\cR_\eff \neq \cR_0$. In general, EDP-convergence allows for
effective dissipation potentials $\cR_\eff$ that inherit properties of
the family $(\cE_\eps)_{\eps>0}$.

The most important feature of the three different notions of EDP-convergence
is that the effective gradient system is \emph{uniquely determined}
through the family $(\sfQ, \cE_\eps, \cR_\eps )$. This is
a much stronger statement than the one obtained by the classical
approach, where the effective or limiting equation is derived first
and then a gradient structure is constructed afterwards. Obviously,
uniqueness cannot guaranteed because one equation  may have several
gradient structures. We also refer
to \cite[Sec.\,3.3.5.2]{Miel16EGCG} where for one family of evolution
equations two different gradient structures are considered such that
the EDP-limit exist and but is different. \EEE

A further interesting observation is that the notion of
EDP-convergence does not involve the solutions of the associated
gradient-flow equation. This may look like an advantage, since
solutions need not be characterized, however typically showing
EDP-convergence is at least as difficult.  Another important feature
is that, \AAAERR under suitable technical assumptions, \EEE 
EDP-convergence automatically implies the convergence of the
corresponding solutions $u^\eps$ of the gradient-flow equations to the
solutions $u$ of the effective equation
\begin{equation}
  \label{eq:EffEqn.Gen}
  0 \in \partial_v \cR_\eff(u(t),\dot u(t)) + \D \cE_0(u(t)) \quad
  \text{for a.a.\ } t \in [0,T].
\end{equation}
\AAAERR The following result gives one possible variant of such a result, see
\cite[Lem.\,2.8]{MiMoPe18?EFED} for another. We do not enforce the condition
$u^\eps(0)\to u(0)$ but only $\cE_\eps(u^\eps(0))\to \cE_0(u(0))$ as well as
the continuity of the limit encoded in the assumption
$u \in \rmW^{1,1}([0,T];X)$. Thus, the result still applies to fast-slow
reaction systems, where jumps at initial time $t=0$ may develop for
$\eps\to 0$, see e.g.\ the example treated in
\cite[Sec.\,2.5]{MiPeSt20EDPCNLRS}. Then, it is important to take into account
that $\lim_{\eps\to 0} u^\eps(0)$ may be different from
$u(0)=\lim_{t\to 0^+} u(t)$.  \EEE
 
\begin{lem}
\label{le:EDPimpliesCvg} Let the assumption of Theorem \ref{thm:EDP}
be satisfied for all $\eps\geq 0$. 
Assume that the gradient systems $(\sfQ;\cE_\eps,\cR_\eps)$
EDP-converge to $(\sfQ,\cE_0,\cR_\eff)$ in one of the three senses of
Definition \ref{def:EDPcvg}, then the following holds. If
$u^\eps:[0,T] \to \sfQ$ are solutions for \eqref{eq:GradFlowEqn} and
$u:[0,T] \to \sfQ$ is such that \AAAERR $ u \in \mathrm W^{1,1}([0,T];X)$,
\[
  \cE_\eps(u^\eps(0))\to \cE_0(u(0)), \ \ \ \AAAERR u^\eps  \rightharpoonup  u
  \text{ in  } \rmL^2([0,T];\sfQ), \ \ \ 
\text{and } \  u^\eps(t)  \to  u(t)  \text{ for all  }t\in {]0,T]}, \EEE 
\]
then \AAAERR  $u $ \EEE is a solution of the \AAAERR effective \EEE 
gradient-flow equation \eqref{eq:EffEqn.Gen}.
\end{lem}
\begin{proof} By Theorem \ref{thm:EDP} we know that the EDB
  \eqref{eq:EDB.gen} holds for $u^\eps$ as solutions for the gradient
  system $(\sfQ,\cE_\eps, \cR_\eps)$, \AAAERR namely 
$\cE_\eps(u^\eps(T))+\fD_\eps(u^\eps) = \cE_\eps(u^\eps(0))$. 

Using $u^\eps(T)\to u(T)$ and $u^\eps \rightharpoonup u$ in $\rmL^2$ 
we have  the liminf estimates 
\[
\cE_0(u(T)) \leq \liminf_{\eps \to 0}  \cE_\eps(u^\eps(T)) \quad
\text{and} \quad \fD_0(u) \leq \liminf_{\eps \to 0}  \fD_\eps(u^\eps).
\]
Together with the assumed convergence of the energies at $t=0$ and the
representation  of $\fD_0$ via $\cR_\eff$, \EEE  we obtain
\begin{align}
\label{eq:EDB.lim}
\cE_0(u(T)) + \int_0^T \!\! \Big(\cR_\eff(u(t),\dot u(t)) +\cR^*_\eff(u(t),
-\D\cE_0(u(t))) \Big) \dd t \leq \cE_0(u(0)). 
\end{align}
Together with \eqref{eq:EDEst.gen} and the EDP in Theorem
\ref{thm:EDP} we see that $u$ solves \eqref{eq:EffEqn.Gen}.
\end{proof}

\section{Gradient structures for linear reaction systems}
\label{SectionGSforLRS}

In this section we discuss several gradient structures for linear reaction
systems satisfying the detailed balance condition. Moreover, following
the theory of Markov processes we define a natural way of tilting such
systems in such a way that a new global equilibrium state $w$
arises. This will show that the entropic gradient structure with
cosh-type dual dissipation plays a distinguished role.

\subsection{A special representation for generators}
\label{su:Generators}

We start from a general linear reaction system with the finite index
space $\cI:=\{1, \dots,   I\}$. On the state space 
$ \sfQ=\mathrm{Prob}(\cI )$ we consider the general linear reaction system 
\begin{equation}
  \label{eq:Evol.Gen}
  \dot c = A c \qquad \text{where } A_{in}\geq 0 \text{ for } i\neq n
\quad \text{and} \quad \sum_{i=1}^I A_{in} =0 \text{ for all }n\in \mathcal I.
\end{equation}
Throughout we assume that there exists a positive equilibrium state $w
\in \sfQ$, i.e.\ $Aw=0$ and $w_i>0$ for all $i\in \cI$. At this stage we
don't need the detailed-balance condition. 

As we later want to change the equilibrium state $w$ (and hence also
the generator $A$) we write $A$ in a specific form, namely 
\begin{equation}
   \label{eq:A.via.K.w}
 \begin{aligned}
 & A= \DDD_w^{1/2} K \DDD_w^{-1/2} - \DDD_b \qquad \text{with
 }K=(\kappa_{in})\in \R^{I\ti I} \text{ and } b\in \R^I \text{ given by}
 \\
 & \kappa_{in}= A_{in} \big( \frac{w_n}{w_i} \big)^{1/2}>0 \text{ for }
 i\neq n, \qquad \kappa_{ii}=0, \ \text{ and }
 \\
 & b_i =-A_{ii} = \sum_{n=1}^I \kappa_{ni}\big(\frac{w_n}{w_i}\big)^{1/2}>0. 
 \end{aligned}
\end{equation}
This representation is useful, because we can keep $K$ fixed, while
varying $w$ to obtain Markov generators $A=A^{w,K}$ such that
$A^{w,K}w=0$. 

Assuming the DBC again, equation \eqref{eq:Evol.Gen} can be written in the
symmetric form
\begin{equation}
  \label{eq:Evol.kappa.w}
  \dot c_n = \sum_{i:\;i\neq n} \kappa_{ni} 
  \Big( \big(\frac{w_n}{w_i}\big)^{1/2}c_i - 
         \big(\frac{w_i}{w_n} \big)^{1/2} c_n\Big) \qquad \text{for
         }n\in \cI. 
\end{equation}

Moreover, we see that $A$ and $w$ satisfies the DBC  $A_{in}w_n=
A_{ni} w_i$ if and only if $K$ is symmetric. Thus, fixing a symmetric $K$
and changing $w$ does automatically generate the DBC for $A^{K,w}$ and
$w$.

\subsection{A general class of gradient structures}
\label{su:GenerClass}

We now assume the DBC $A\DDD_w =(A\DDD_w)^*$ or equivalently
$K=K^*$ in \eqref{eq:A.via.K.w} and discuss a general class of
gradient structures for \eqref{eq:Evol.Gen} following the 
general approach in \cite[Sec.\,2.5]{MaaMie18?MCRD}.

Let $\Phi:{[0,\infty[}\to {[0,\infty[}$ and $\Psi_{in}:\R\to
{[0,\infty[}$ for $1\leq i<n\leq I$ be lower semi-continuous and
strictly convex $\mathrm C^2$ functions such that $\Psi_{in}(0)=0$ and
$\Psi''_{in}(0)>0$.  We search for a gradient system $(\sfQ,\cE,\cR^*)$
with an energy functional $\cE$ and a dual dissipation potential in
the form
\begin{align*}
  \cE(c) = \sum_{i=1}^I w_i\,\Phi\big(\frac{c_i}{w_i} \big)
 \qquad \text{and} \qquad
  \cR^*(c,\xi) = \sum_{i=1}^{I-1} \sum_{n=i+1}^{I}
               a_{in}(c)\,\Psi_{in}(\xi_i {-}\xi_n),
\end{align*}
where the coefficient functions $a_{in}$ must be chosen appropriately,
but need to be nonnegative to guarantee that $\cR^*(c,\cdot)$ is a
dissipation potential.
 
With $\partial_{\xi_n} \cR^*(c,\xi) = \sum_{k=n+1}^I a_{nk}(c)
\Psi'_{nk} (\xi_n {-}\xi_k) - \sum_{i=1}^{n-1} a_{in}(c) \Psi'_{in} (
\xi_i{-}\xi_n) $ and $\D\cE(c)=\big(\Phi'(\frac{c_k}{w_k})\big)_k$
we find the relation
\[
\partial_{\xi_n} \cR^*(c,{-}\D\cE(c))
 = \sum_{i=n+1}^I a_{ni}(c)\Psi'_{ni} \Big(\Phi'\big(\frac{c_i}{w_i} \big)
 -\Phi'\big(\frac{c_n}{w_n}\big)\Big) 
 - \sum_{i=1}^{n-1} a_{in}(c) \Psi'_{in}\Big(\Phi'\big(\frac{c_n}{w_n}\big) 
        -\Phi'\big(\frac{c_i}{w_i}\big) \Big).  
\]
Thus, the equations $\dot c_n=\partial_{\xi_n} \cR^*(c,{-}\D\cE(c))$
are the same as in \eqref{eq:Evol.kappa.w}, provided we choose the
coefficient functions $a_{in}$ as
\begin{align}
 \label{eq:a.in.choice}
  a_{ni}(c) :=  \frac{\kappa_{ni}\,\sqrt{w_nw_i} \,\big(\frac{c_i}{w_i} -
    \frac{c_n}{w_n})} {\,\Psi'_{ni}\big(\Phi'(\frac{c_i}{w_i}) -
    \Phi'(\frac{c_n}{w_n})\big) \,} \text{ for }\tfrac{c_i}{w_i}\neq
  \tfrac{c_n}{w_n} \  \text{ and }\   a_{ni}(c) :=
  \;\frac{\kappa_{ni}\,\sqrt{w_nw_i} } {\,\Psi''_{ni}(0) \Phi''(\frac{c_i}{w_i})\,} 
  \text{ for }\tfrac{c_i}{w_i}= \tfrac{c_n}{w_n}
\end{align}
and exploit the DBC $\kappa_{in}=\kappa_{ni}$. We also emphasize that
$\Phi'$ is strictly increasing such that $\frac{c_i}{w_i} -
\frac{c_n}{w_n}$ and $\Phi' (\frac{c_i}{w_i}) - \Phi'(
\frac{c_n}{w_n})$ always have the same sign. Since $\Psi'(\zeta)$ and
$\zeta$ also always have the same sign, we conclude that
$a_{in}(c)\geq 0$ as desired for dissipation potentials.

As the choice of entropy functional density $\Phi$ and of the dual
dissipation potentials $\Psi_{in}$ is general quite arbitrary we see
that we can generate a whole zoo of different gradient structures for
\eqref{eq:Evol.Gen} or \eqref{eq:Evol.kappa.w}.
The following choices relate to situation where all $\Psi_{in}$ are
given by one function $\Psi$, but more general cases are possible. 

From the construction it is clear that $\cR^*$ is linear in the
generator $A$, i.e.\ if $A= A^1 + A^2$ and the equilibrium $w$ is
fixed, then $\cR^*=\cR^*_{A^1} + \cR^*_{A^2}$ where $\cR^*_{A^m}$ is
constructed as above.

\subsection{Some specific gradient structures for linear reaction systems}
\label{su:SomeSpecialGS}

We now realize special choices for the general gradient structures in
the previous subsection. These choices are singled out because they
lead to natural entropy functionals and relatively simple coefficient
functions $a_{in}$ in \eqref{eq:a.in.choice}.

\subsubsection{Quadratic energy and dissipation}
\label{suu:QuadraticGS}

The quadratic gradient structure is given by quadratic energy and
dissipation, i.e. 
\[
 \Phi_\text{quad}(\varrho )=\frac 1 2 \varrho^2
   \quad \text{ and } \quad
 \Psi_\text{quad}(\zeta)=\frac 1 2 \zeta^2.
\] 
The coefficient functions are constant and read $a_{in}(c) =
\kappa_{in}\sqrt{w_iw_n}$. Thus, we find 
\[
\cE_\text{quad}(c)= \frac 1 2 \sum_{i=1}^I \frac {c_i^2}{w_i} \quad 
\text{ and } \cR^*_\text{quad}(c,\xi) = \frac 1 2 \sum_{i=1}^{I-1}
     \sum_{n=i+1}^{I} \kappa_{in}\sqrt{w_iw_n} (\xi_i{-}\xi_n)^2 
 =\frac12\la \xi , \mathbb K_\text{quad}\xi\ra. 
\]
In this case the dual dissipation functional does not
depend on the concentration $c \in \sfQ$, which means that the equation
$\dot c= Ac= -\mathbb K \D \cE(c)$ can be treated as self-adjoint linear
evolution problem in the Hilbert space with the norm induced by
$\cR$. This leads to the classical Hilbert space approach for
reversible Markov operators.

\subsubsection{Boltzmann entropy and quadratic dissipation}
\label{suu:QuadrEntrGS}
The \textit{quadratic-entropic} gradient structure is defined by the
choices 
\[
\Phi_\text{Boltzmann}(\varrho )=\LB(\varrho):= \varrho \log \varrho -
\varrho +1  \quad \text{ and } \quad
\Psi_\text{quad}(\zeta)=\frac 1 2 \zeta^2.
\] 
This gradient structure for was first introduced in \cite{Miel11GSRD,
  Maas11GFEF, ErbMaa12RCFM, CHLZ12FPEF, Miel13GCRE} as a possible 
generalization of Otto's gradient structure for the Fokker-Planck and
more general diffusion equations
equation, cf.\ \cite{JoKiOt98VFFP, Otto01GDEE}. However, similar
structures also appear earlier in the physics literature, see e.g.\
\cite[Eqn.\,(113)]{OttGrm97DTCF2} 

The associated entropy is Boltzmann's relative entropy and, using the
logarithmic mean $\Lambda(a, b) = \int_0^1a^sb^{1-s} \dd
s=\frac{a-b}{\log a-\log b}$, the dual dissipation potential $\cR^*$
reads
\begin{align*}
  \cE_\text{Bz}(c):=\sum_{i=1}^I w_i\,\LB \big(\frac{c_i}{w_i}\big)
  \quad \text{and} \quad 
  \cR^*(c,\xi) = \frac 1 2 \sum_{i=1}^{I-1} \sum_{n=i+1}^{I}
  \kappa_{in}\, \sqrt{w_iw_n} \, \Lambda 
      \big(\frac{c_i}{w_i}, \frac{c_n}{w_n}\big)\,(\xi_i{-}\xi_n)^2.
\end{align*}
Again $\cR^*$ is quadratic in $\xi$ but now also depends nontrivially
on $c\in \sfQ$, viz.\ $\cR^*(c,\xi)=\frac12 \la \xi,\mathbb
K_\text{Bz}(c)\xi\ra$.  This means that $\sfQ$ can be equipped with the Riemannian
metric induced by $\cR$, see \cite{Maas11GFEF}.  

Note that $\mathbb K_\text{Bz}(w)= \mathbb K_\text{quad}$ and
$\cE_\text{quad}(c)=\frac12 \D^2 \cE_\text{Bz}(w)[c,c]$, which is the
desired compatibility under linearization at $c=w$.

\subsubsection{Boltzmann entropy and cosh-type dissipation}
\label{suu:Boltz.cosh}
The following, so-called \textit{entropic cosh-type gradient
  structure}, was derived via a large-deviation principle from an
interacting particle system in \cite{MiPeRe14RGFL, MPPR17NETP}. We
refer to Marcellin's PhD thesis \cite{Marc15CECP} from 1915 for a
historical, first physical derivation of exponential kinetic relations
in the context of Boltzmann statistics. Only little of this important
result penetrated into the main stream thermomechanical modeling of
reaction systems, see  \cite[Item iii on p.\,77 and
eqn.\,(69)]{Grme10MENT} for a discussion.

For this gradient structure the choices are  
\[
\Phi_\text{Boltzmann}(\varrho )=\LB(\varrho):= \varrho \log \varrho -
\varrho +1  \quad \text{ and } \quad
\Psi_\text{cosh}(\zeta)=\sfC^*(\zeta):= 4
\cosh\big(\frac\zeta2\big) - 4, 
\] 
giving Boltzmann's relative entropy $\cE_\text{Bz}$ and the
cosh-type dual dissipation potential: 
\begin{equation}
  \label{eq:EBz.cRcosh}
  \cE_\text{Bz}(c):=\sum_{i=1}^I w_i\,\LB \big(\frac{c_i}{w_i}\big)
  \quad \text{and} \quad 
  \cR_\text{cosh}^*(c,\xi) = \sum_{i=1}^{I-1}\sum_{n=i+1}^{I}
  \kappa_{in}\,\sqrt{c_ic_n} \,\sfC^*(\xi_i{-}\xi_n).
\end{equation}
The especially simple form of the coefficient functions arises from
the interaction of the $\cosh$ function with the the Boltzmann
function $\LB$, namely 
\begin{align*}
  \sfC^{*\prime}\big(\LB'(p) - \LB'(q)\big)= 2 \sinh
  \big(\log\sqrt{p/q}\:\big) = \sqrt{p/q}  - \sqrt{q/p} = 
  \frac{p-q}{\sqrt{pq}}.
\end{align*}
With this we easily find the simple formula $a_{in}(c)=\kappa_{in}
\sqrt{c_ic_n}$.

Because of the close connection between the cosh-type function
$\sfC^*$ and the Boltzmann function $\LB$, it is obvious
that using $\sfC^*$ means that we also use the Boltzmann
entropy. Hence, it will not lead to confusion if we simply call
$(\sfQ,\cE_\text{Bz}, \cR_\text{cosh})$ the \emph{cosh gradient
  structure}. 

Again, the quadratic gradient structure in Section
\ref{suu:QuadraticGS} is obtained by linearization:
\[
\cE_\text{quad}(c)=\frac12 \D^2 \cE_\text{Bz}(w)[c,c]
  \quad \text{and} \quad 
\mathbb K_\text{quad} = \D_\xi^2 \cR_\text{cosh}^*(w,0).
\]

\subsection{Tilting of Markov processes}
\label{su:Tilting}

Tilting, also called exponential tilting, is a standard procedure in
stochastics (in particular in the theory of large deviations) 
to change the dynamics of a Markov process in a controlled
way. In particular, the equilibrium measure $w$ is changed into
another one, let us say $\wt w$. For more motivation and theory we
refer to \cite{MiMoPe18?EFED} and the references therein. 

Defining two entropy functionals, namely the Boltzmann entropies for
$w$ and $\wt w$,
\[
\cE_\text{Bz}(c)=\sum_{i=1}^I w_i \,\LB\big( \frac{c_i}{w_i}\big) 
  \quad \text{and} \quad 
\wt\cE_\text{Bz}(c)=\sum_{i=1}^I \wt w_i \,\LB\big( \frac{c_i}{\wt w_i}\big) 
\]
the special structure of $\LB$ leads to  the relation 
\[
\wt \cE_\text{Bz} (c) = \cE_\text{Bz}(c)- \la \eta, c\ra \quad
\text{with } \eta =\big( \log(w_i/\wt w_i)\big)_{i\in \cI}\,. 
\]
Thus, we see that a change of the equilibrium measure leads to a tilt
in the sense of \eqref{eq:cE.tilt} for the entropy.  Moreover, for
every tilt $\eta\in X^*$ there is a unique new equilibrium state
$w^\eta$, namely the minimizer of $c\mapsto \cE^\eta(c) =
\cE_\text{Bz}(c)-\la \eta,c\ra$. We easily find
\[
w^\eta_i = \frac1Z \e^{-\eta_i} w_i \quad \text{with } Z=\sum_{n=1}^I
\e^{-\eta_n} w_n. 
\]   
This explains the name `exponential tilting'. 

For a time-dependent linear reaction systems the tilting is defined in
a consistent way, namely using the representation
\eqref{eq:A.via.K.w}. Given $\dot c =Ac$ with positive equilibrium $w$
and a tilt $\eta$ we first construct the equilibrium $w^\eta$ and
then, using $K=(\kappa_{in})$ from \eqref{eq:A.via.K.w}, we define the
evolution
\begin{equation}
  \label{eq:ODE.tilted}
  \dot c = A^\eta c \quad \text{with } A^\eta :=\DDD_{w^\eta}^{1/2} K
\DDD_{w^\eta}^{-1/2} -\DDD_{b^\eta}.
\end{equation}

One of the important observations in \cite{MiMoPe18?EFED} is that the
cosh gradient structure is invariant under
tilting, i.e.\ the dissipation potential does not change if the
Boltzmann entropy is tilted. This can now be formulated as follows:
\begin{equation}
  \label{eq:rel.tilt.GS}
  A^\eta c = \D_\xi \cR^*_\text{cosh}\big(c, -\D\cE^\eta(c) \big).
\end{equation}
This relation can easily checked by noting that \eqref{eq:ODE.tilted} has the
form \eqref{eq:Evol.kappa.w}, where now $w$ is replaced by $w^\eta$. 
But $\cE^\eta$ is exactly the relative entropy with respect to
$w^\eta$ such that the results in Section \ref{suu:Boltz.cosh} yield
identity \eqref{eq:rel.tilt.GS}.  

Using the formula \eqref{eq:a.in.choice} for $a_{in}(c)$ we can find
all possible gradient structures in terms of $\Phi$ and $\Psi_{in}$
such that the $a_{in}(c)$ is independent for $w$. The result shows
that, up to a trivial scaling, the only tilt-invariant gradient structures in
the form of Section \ref{su:GenerClass} are given by the cosh gradient
structure. Indeed, in \cite{MiPeRe14RGFL} the case $\gamma=1/2$ is
obtained from the theory of large deviations. 

\begin{prop}[Characterization of tilt-invariant gradient structures]
\label{pr:Tilt.Inv.GS}
If $\Phi$ and $\Psi_{in}$ are such that $a_{in}$ in
\eqref{eq:a.in.choice} is independent of $w$, then there exists
$\varphi_0,\varphi_1\in \R$ and $\psi_{in},\gamma>0$ such that 
\[
\Phi(c)=\gamma \LB(c) + \varphi_0 + \varphi_1c \quad \text{and} 
\quad \Psi_{in}(\zeta) = \gamma\, \psi_{in}\, 
 \sfC^*\big(\frac \zeta \gamma \big).
\]
In particular, we always obtain $a_{in}(c)=\frac{\kappa_{in}}{\psi_{in}}
\sqrt{c_ic_n}$. Since $\psi_{in}$ can be integrated into $\kappa_{in}$,
all tilt-invariant gradient structures are given by 
scaled cosh gradient structures 
\[
\cE(c)= \gamma \,\cE_\text{Bz}(c) + \varphi_0I+\varphi_1 
  \quad \text{and} \quad 
\cR^*(c,\xi)=\gamma \, \cR^*_\mathrm{cosh}(c,\frac1\gamma \xi).
\]
\end{prop}
\begin{proof}
We rewrite 
$a_{in}$ in the form 
\[
 a_{in}(c) =  \kappa_{in}\sqrt{c_ic_n} \: 
   \frac{\varrho_i - \varrho_n} {\sqrt{\varrho_i\varrho_n} \,
  \Psi'_{ni}\big(\Phi'(\varrho_i) - \Phi'(\varrho_n)\big)}, \ \text{ where  }
  \varrho_k=\frac{c_k}{w_k} 
\] 
Because the expression has to be independent of $w_i$ and $w_n$ for
all $c,w\in \sfQ$, the fraction involving $\varrho_i$ and $\varrho_n$
has to be a constant, which we set $1/\psi_{in}$ , i.e.\
\[
\text{(i) } \ \Phi'(\varrho_i) - \Phi'(\varrho_n) =
G\big(\frac{\varrho_i}{\varrho_n} \big), \quad \text{ (ii) } \
G(\sigma)= \big(\Psi'_{in} \big)^{-1} \Big(\psi_{in}  
  \big( \sqrt\sigma - \frac1{\sqrt\sigma}\big) \Big)  .
\]

Setting $r_k=\log \varrho_k$, $f(r)=\Phi'(\e^r)$, and
$g(s)=G(\e^s)$ in (i), we arrive at the relation
\[
f(r_i)-f(r_n) = g(r_i{-}r_n)\quad  \text{for all } r_i,r_n\in \R. 
\]
As $f$ and $g$ are continuous the only solutions of this functional 
relation are $f(r)=\varphi_1 + \gamma r$ and $g(s)=\gamma s$ with
$\varphi_1,\gamma\in \R$. This implies $\Phi'(\varrho)=\varphi_1+
\gamma \log \varrho$ and, hence, 
$\Phi(\varrho)=\varphi_0+\varphi_1 \varrho + \gamma
\LB(\varrho)$. Strict convexity of $\Phi$ leads to the restriction
$\gamma>0$. 
 
Solving (ii) with $G(\sigma)= \gamma \log \sigma=:\zeta$ yields 
\[
\Psi'_{in}(\zeta) = \psi_{in} \big( \e^{\zeta/(2\gamma)} -
\e^{-\zeta/(2\gamma)} \big) = \psi_{in} \,2 \sinh\big(
\frac{\zeta}{2\gamma}\big)= \psi_{in}\, \mathsf
C^*{}'\big(\frac{\zeta}\gamma\big).
\]
Because of $\Psi_{in}(0)=0$ this determines $\Psi_{in}$ uniquely, and
the result is established.  
\end{proof}

We also refer to \cite{HeKaSt20DSFP} for the connections of the
cosh gradient structure to the SQRA-discretization
scheme for drift-diffusion systems.

\section{EDP-convergence and the effective gradient structure}
\label{SectionMainTheorem}

In this section we fully concentrate on the cosh gradient
structure, because only this gradient structure allows
to prove EDP convergence with tilting. 

Our energy functionals $\cE_\eps$ are the relative Boltzmann
entropies, while the dual dissipation potentials $\cR^*_\eps$ is the
sum of a slow and a fast part: 
\begin{align*}
 &\cE_\eps(c) = \sum_{i=1}^I w_i^\eps\,\LB \big(\frac{c_i}{w_i^\eps} \big)
  \quad \text{and} \quad 
 \cR^*_\eps( c,\xi) = \cR^*_{S,\eps}(c,\xi) + \frac1\eps\cR^*_{F,\eps}(c,\xi),
 \text{ where }
\\
&\cR^*_{Z,\eps} (c,\xi):= \sum_{i=1}^{I-1} \sum_{n=i+1}^I
\kappa_{in}^{Z,\eps} \sqrt{c_i c_n} \, \sfC^*(\xi_i {-} \xi_n)
\quad \text{with }   \kappa_{in}^{Z,\eps} = A^Z_{in} \sqrt{w^\eps_n/w^\eps_i}
\quad\text{and } Z\in \{S,F\}.
\end{align*}
Here, the $\eps$-dependencies of the coefficients
$\kappa^{S,\eps}_{in} $ and $\kappa^{F,\eps}_{in}$ is trivial in the
sense that the limits for $\eps\to 0$ exist. The really important term
is the factor $1/\eps$ in front of $\cR^*_{F,\eps}\,$. 

The structure of this section is as follows. In Section
\ref{su:MainTheorem} we present the main results concerning the
$\Gamma$-convergence of $\cE_\eps$ and $\fD_\eps$ which then imply
the EDP-convergence with tilting of $(\sfQ,\cE_\eps,\cR_\eps)$ to the
limit system $(\sfQ,\cE,\cR_\eff)$. In Section \ref{SectionLimitingGS}
we show that this provides a gradient structure for the limit
equation \eqref{eq:LimitEqn}, and moreover that we obtain the natural
cosh gradient structure $(\hat \sfQ,\hat\cE,\hat\cR)$ for the coarse-grained
equation \eqref{eq:CoarseGraEqn}.    

The remaining part of this section then provides the proof of the
convergence $\fD_\eps \Moscolim \fD_0$, namely the a priori estimates
in Section \ref{su:BoundsComp}, the liminf estimate in Section
\ref{su:LiminfEst}, and the construction of recovery sequences in
Section \ref{su:RecoverySeq}.

\subsection{Main theorem on EDP-convergence }
\label{su:MainTheorem}

We now study the limit for $\eps\to 0$ of the family
of gradient systems $\big((\sfQ,\cE_\eps,\cR_\eps^*)\big)_{\eps>0}$ by
showing EDP-convergence with tilting for a suitable limit.  

As a first, and trivial result we state the Mosco convergence of the
energies, which follows immediately from our assumption
\eqref{eq:cond.A2}, i.e.\ $w^\eps\rightarrow w^0$. 

\begin{prop}\label{pr:GCvg.calE}
  On $\sfQ=\mathrm{Prob}(\cI )$, we have the uniform convergence
  $\cE_\eps \rightarrow \cE_0$, where $\cE_0(c) = \sum_{i=1}^I w_i^0\,
  \LB (c_i/w_i^0)$. In particular, we have $\cE_\eps \Moscolim
  \cE_0$ on $X$.
\end{prop}

To have a proper functional analytic setting we let 
\[
\rmL^2([0,T];\sfQ)=\bigset{ c\in  \rmL^2([0,T];\R^I)}{c(t)\in \sfQ \text{
    a.e.\ in }[0,T]}
\]
and use the weak and strong topology induced by $\rmL^2([0,T];\R^I)$. 
The dissipation functional $\fD_\eps$ is now defined via
\begin{align*}
  \mathfrak D_\eps (c) := 
 \begin{cases}
  \int_0^T \!\! \big(\cR_\eps(c,\dot c) +
  \cR^*_\eps(c, {-} \D\cE_\eps(c)) \big) \dd t &
  \text{for } c \in \rmW^{1,1}([0,T];\sfQ),\\
 \infty& \text{otherwise on } \rmL^2([0,T];\sfQ),
 \end{cases}
\end{align*}
where $\cR_\eps(c,\cdot)$ is defined implicitly as Legendre transform of
$\cR^*_\eps(c,\cdot)$. To see that $\fD_\eps$ is well defined, we
derive suitable properties for $\cR_\eps$.

\begin{prop}[Properties of $\cR_\eps$] \label{pr:Proper.cReps}
Let $\cR_\eps:\sfQ\ti X\to [0,\infty] $ be defined by $\cR_\eps(c,\cdot)
= \big( \cR_\eps^*(c,\cdot)\big)^*$. Then, $\cR_\eps:\sfQ\ti X\to
[0,\infty] $ \EEE is lower semicontinuous and jointly convex. 
\end{prop}
\begin{proof}
Since $(c_i,c_n)\mapsto \sqrt{c_ic_n}$ is concave and $\xi\mapsto
\sfC(\xi_i{-}\xi_n)$ is convex, the mapping $\cR^*:\sfQ\ti X^*\to
[0,\infty]$ is concave-convex and thus its partial conjugate is
convex in $(c,v)$. 

For the lower semicontinuity consider $(c_k,v_k)\to (c,v)$. Then, for
all $\delta>0$ there exist $\xi_\delta$ with $
\cR_\eps(c,v)\leq \la \xi_\delta,v\ra -\cR^*_\eps(c_k,\xi_\delta) + \delta$. 
The definition of the Legendre transform yields
\[
\cR_\eps(c_k,v_k)\geq \la \xi_\delta,v_k\ra
-\cR^*_\eps(c_k,\xi_\delta) \overset{k\to\infty}\to 
\la \xi_\delta,v\ra -\cR^*_\eps(c,\xi_\delta) \geq \cR_\eps(c,v)-\delta,
\]
where we used the continuity of $c\mapsto \cR^*_\eps(c,\xi)$. 
Since $\delta>0$ was arbitrary, we find $\liminf_{k\to \infty}
\cR_\eps(c_k,v_k) \geq \cR_\eps(c,v)$ as desired. 
\end{proof}

To formulate the main $\Gamma$-convergence result for $\fD_\eps$ we
define the effective dissipation $\cR^*_\eff$ beforehand. It can be
understood as the formal limit of $\cR^*_\eps$ when taking $\eps\to
0$. The slow part $\cR^*_{S,\eps}$ simply converges to its limit 
\begin{align*}
   \cR_S^*(c, \xi) := \sum_{i=1}^{I-1} \sum_{n=i+1}^I
   \kappa^{S,0}_{in} \, \sqrt{c_ic_n} \,\sfC^*(\xi_{i} - \xi_{j})
   \quad \text{ with }
\kappa^{S,0}_{in} =A_{in}^S\sqrt{w_n^0/w_i^0}=\lim_{\eps\to 0} \kappa_{in}^{S,\eps}. 
\end{align*}
For the fast part $\frac1\eps \cR^*_{F,\eps}$ we obtain blow up,
except for those $\xi$ that lie in the subspace that is not affected
by fast reactions. For this we set
\[
\Xi=M^*Y^*=\mathrm{range}(M^*) = \mathrm{kernel}(M)^\perp :=
\bigset{\xi \in X^*}{ \la \xi,v\ra =0 \text{ for all } v \in
  \mathrm{kernel} (M) }.  
\] 
and observe that by construction for all $\eps>0$ we have 
\begin{equation}
  \label{eq:cRF.Xi}
\cR^*_{F,\eps}(c,\xi) =0 \quad \text{ for all }\xi \in
\Xi.   
\end{equation}
Indeed, $\cR^*_{F,\eps}(c,\xi) $ contains $\sfC^*(\xi_i{-}\xi_n)$ with
a positive prefactor only if $i\sim_F n$, while $\xi\in \Xi$ implies
$\xi_i=\xi_n$ in that case.  Together we set
\begin{equation}
  \label{eq:cR*eff}
  \cR^*_\eff(c,\xi) := \cR^*_S(c,\xi) + \chi_{\Xi}(\xi), \ \text{ where }
  \chi_A(a)=\begin{cases} 0&\text{for }a\in A,\\ \infty& \text{for
    }a\not\in A. \end{cases} 
\end{equation}
The dual dissipation potential $\cR_\eff^*$ consists of two
terms: The first term $\cR_S^*$ contains the information of the slow
reactions in the limit $\eps\rightarrow 0$. The second term
$\chi_{\Xi}$ restricts the vector of chemical potentials $\xi = \D
\cE_0(c)$ exactly in such a way that the microscopic equilibria of the
fast reactions holds, i.e.\ $A^Fc=0$ or equivalently $Pc=c$, see
below.

Because of this constraint, it is actually irrelevant how
$\cR^*_\eff(c,\cdot):\Xi \to [0,\infty]$ is defined for $c \not\in
\sfQ_\mathrm{eq}=\sfQ\cap PX$. 

We note that $\cR^*_\eps(c,\cdot)$ has a Mosco limit
$\cR^*_0(c,\cdot)$ that is not necessarily equal to
$\cR^*_\eff(c,\cdot)$. For $c$ on the boundary of $\sfQ$, where some
of the $c_i$ are $0$,  we may have
$\cR^*_{F,\eps}(c,\xi)=0$ for all $\xi$, which implies
$\cR^*_0(c,\xi)=\cR^*_{S}(c,\xi)$ for these $c$ and all $\xi \in
\R^I$. However, the $\Gamma$-limit of $\fD_\eps$ yields
$\cR^*_\eff \geq \cR^*_0$. 

\begin{thm}[Mosco convergence of $\fD_\eps$]
\label{thm:MoscoCvg.frakD}
  On $\rmL^2([0,T];\sfQ)$ we have $\mathfrak D_\eps
  \Moscolim \fD_0$ with
\begin{align}
  \label{eq:frakD0}
 \fD_0(c) := \begin{cases}
           \text{\larger $\int_0^T$} \big(\cR_\eff(c, \dot c) + 
           \cR_\eff^*(c, - \D \cE_0(c))\big) \dd t
           &\text{for } c\in \rmW^{1,1}([0,T];\sfQ), \\
            \infty& \text{otherwise in }\rmL^2([0,T];\sfQ),
                    \end{cases}
\end{align}
where $\cR^*_\eff$ is given in \eqref{eq:cR*eff} and leads to the
primal dissipation potential 
\begin{align*}
 \cR_\eff(c,v) = \inf\bigset{\cR_S(c, z)}{z\in\R^I \text{ with } Mz=
   Mv} \quad \text{for all }c\in \sfQ_\mathrm{eq}=P\sfQ .
\end{align*}
\end{thm}

The proof of this theorem is the main part of this section and will be
given in Sections \ref{su:BoundsComp} to \ref{su:RecoverySeq}. Now, we want to use the above result to conclude the EDP-convergence
with tilting. For this result, it is essential to study the dependence
of the limit $\fD_0$ on the limit equilibrium measure $w^0$. One
the one hand, $\cE_0(c)$ is the relative Boltzmann entropy of $c$ with
respect to $w^0$, which provides a simple and well-behaved dependence
on $w^0$. On the other hand, $\cR^*_\eff$ is given through
$\cR^*_S$ and $\chi_\Xi$. The former only depends on
$(\kappa^{S,0}_{in})_{i,n\in \cI}$ and the latter depends only on $M \in
\{0,1\}^{J\ti I}$. Thus, there is no dependence on $w^0$ at all. The
proof relies on the fact that the two processes of (i) tilting with 
driving forces $\eta$ and of (ii) taking the limit $\eps\to 0$
commute.

\begin{thm}[EDP-convergence with tilting]
\label{thm:EDPtilt}
  The gradient systems $(\sfQ,\cE_\eps,\cR_\eps)$ EDP-converge with
  tilting to the limit gradient structure $ (\sfQ,\cE_0,\cR_\eff)$.

  The closure of the domain of the limit gradient system in the
  sense of \eqref{eq:Dom.GS} is $\sfQ_\mathrm{eq}$.
\end{thm}
\begin{proof} Proposition \ref{pr:GCvg.calE} and Theorem
\ref{thm:MoscoCvg.frakD} already provide the simple EDP convergence
$(\sfQ,\cE_\eps, \cR^*_\eps) \EDPto (\sfQ,\cE_0,\cR^*_\eff)$.  The
domain is restricted by the conditions (i) that $\D\cE_0(c)$ exists,
which means that $ c_i>0$ for all $i$, and (ii) that $\D\cE_0(c)$ lies
in the domain of $\partial_\xi \cR_\eff^*(c,\,\cdot\,)$. \EEE The
latter condition is equivalent to $\D \cE_0(c)\in \Xi$ or equivalently
$c\in X_\mathrm{eq}$.

For the tilted energies $\cE^\eta_\eps= \cE_\eps -\la \eta,\cdot\ra $
we obviously have $\cE^\eta_\eps \Moscolim \cE_0^\eta$. We can now
apply Theorem \ref{thm:MoscoCvg.frakD} once again for $\fD^\eta_\eps$.
Using the fact that $\cE^\eta$ is again a relative Boltzmann entropy
with respect to the exponentially tilted equilibrium state
$w^{\eta,\eps}$ that satisfies $w^{\eta,\eps} \to w^{\eta,0}$. Thus,
the Mosco limit $\fD^\eta_0$ of $\fD^\eta_\eps$ again exists and has
the same form as $\fD_0$ in \eqref{eq:frakD0}, but with $\D\cE_0(c)$
replaced by $\D\cE(c)-\eta$. In particular, $\cR_\eff$ remains
unchanged and EDP-convergence with tilting is established.
\end{proof}

\subsection{The limit and the coarse-grained gradient structure}
\label{SectionLimitingGS}

Before going into the proof of Theorem \ref{thm:MoscoCvg.frakD} we
connect the limit gradient systems with the limit equation
\eqref{eq:LimitEqn}. The gradient-flow equation for the limit gradient
systems reads
\begin{equation}
  \label{eq:GradFlowLim}
  \dot c \in \partial_\xi \cR^*_\eff(c, {-}\D \cE_0(c)) \quad \text{a.e.\ on
  } [0,T]. 
\end{equation}
Since $\cR^*_\eff$ is no longer smooth, we use the set-valued convex
subdifferential $\partial_\xi$ that satisfies, because of the continuity
of $\cR^*_S$, the sum rule 
\[
\partial_\xi \cR^*_\eff(c,\xi)= \D_\xi \cR^*_S(c,\xi) + \partial
\chi_\Xi (\xi) \quad \text{with } \partial\chi_\Xi(\xi)
= \begin{cases}
   \mathrm{kernel}(M) & \text{for }\xi \in \Xi, \\
    \emptyset & \text{for }\xi \not\in \Xi,
  \end{cases}
\]
where we used the relation 
$\Xi= \mathrm{range}(M^*)=\mathrm{kernel}(M)^\perp$. 

On the one hand, \eqref{eq:GradFlowLim} implies that $\D\cE_0(c)\in \Xi$ for
a.a.\ $t\in [0,T]$. Recalling that the rows of $M\in \{0,1\}^{J\ti I}$
consists of vectors having the entry $1$ in exactly one equivalent class
$\alpha(j)\subset \cI$ for $\sim_F$ and $0$ else, we have 
\[
\Xi=\mathrm{range}(M^*)=\bigset{\xi\in \R^I}{ \forall\,j\in \cJ\ 
  \forall\, i_1,i_2\in \alpha(j): \ \xi_{i_1}=\xi_{i_2} }
\]
we conclude 
\[
\D\cE_0(c) \in \Xi \ \Longleftrightarrow \ 
\forall\,j\in \cJ\ 
  \forall\, i_1,i_2\in \alpha(j): \ \frac{c_{i_1}}{w^0_{i_1}}= 
 \frac{c_{i_2}}{w^0_{i_2}} \ \Longleftrightarrow \  c \in
 X_\mathrm{eq} \ \Longleftrightarrow \ A^F c=0.
\]
One the other hand, by construction of the gradient structure the
term $\D_\xi \cR^*_S(c,{-}\D \cE_0(c))$ generates exactly
the term $A^S c$. Thus, \eqref{eq:GradFlowLim} is equivalent to 
\begin{equation}
  \label{eq:GradFlow-2}
  \dot c(t) \in A^S c(t) + \mathrm{kernel}(M),\quad  A^Fc(t)=0
  \qquad \text{a.e.\ on } [0,T]. 
\end{equation}
Applying $M$ to the first equation gives
the limit equation \eqref{eq:LimitEqn} and the following result.

\begin{prop}[Gradient structure for limit equation]
 \label{pr:GS4LimEqn}  
The limit equation \eqref{eq:LimitEqn} is the
gradient-flow equation generated by the limit gradient system
$(\sfQ,\cE_0,\cR^*_\eff)$. 
\end{prop}

As a last step, we show that the gradient structure for the limit
equation also provides a gradient structure for the coarse gradient
equation \eqref{eq:CoarseGraEqn} $\dot{\hat c} = MA^SN\hat c$ for the
coarse-grained states $\hat c = Mc \in \hat\sfQ$. For this we exploit
the special relations derived for coarse graining via $M:X\to
Y$ and reconstruction via $N:Y\to X$. 

\begin{thm}[Gradient structure for coarse-grained equation] 
\label{thm:GS4CoarseGraEqn}
The coarse-grained equation \eqref{eq:CoarseGraEqn} (viz.\ $\dot{\hat
  c}= MA^SN\hat c$) is the
gradient-flow equation generated by the coarse-grained gradient system 
$(\hat \sfQ,\hat\cE, \hat\cR)$ given by 
\[
\hat\cE(\hat c) = \cE_0(N\hat c) = \mathcal H_J(\hat c| \hat w) 
\quad \text{and} \quad  \hat \cR(\hat c, \hat v)
 = \cR_\eff( N \hat c, N \hat v). 
\]
Moreover, we have $\hat\cR^*(\hat c,\hat\xi)= \cR^*_\eff(N\hat c, M^*\hat\xi) =
\cR^*_S (N\hat c, M^*\hat\xi) $.
\end{thm}

This result can be seen as an exact coarse graining in the sense of
the formal approach developed in \cite[Sec.\,6.1]{MaaMie18?MCRD}.

Before giving the proof of this result we want to highlight its
relevance. First, we emphasize that the coarse-grained equation is
again a linear reaction system, now in $\R^J$, i.e.\ the master
equation for a Markov process on $\cJ=\{1,\ldots,J\}$. Second, the
coarse-grained energy functional is again the relative Boltzmann
entropy, now with respect to the coarse-grained equilibrium $\hat w =
Mw^0$. Third, the coarse-grained dual dissipation potential is again
given in terms of the function $\sfC^*$, i.e.\ the coarse-grained
gradient system is again of cosh-type. In summary, the coarse-grained
gradient structure $(\hat \sfQ,\hat\cE, \hat\cR)$ is again a cosh
gradient structure, see Proposition
\ref{pr:kappa.hat.kappa} below.

We refer to \cite[Sec.\,3.3]{LMPR17MOGG} for an example that shows
that other gradient structures may not be stable under
EDP-convergence.  All these results rely on the special properties of
$M$ and $N$ developed in Section \ref{su:PropCG+Recov}.  In
particular, we use that the projection $P=NM:X\to X$ is a Markov
operator, i.e.\ it maps $\sfQ$ onto itself.
 
\noindent
\begin{proof}[Proof of Theorem \ref{thm:GS4CoarseGraEqn}] \mbox{}

\underline{Step 1: $\hat\cE$ is a relative entropy.} We use the
special form $N=\DDD_{w^0} M^* \DDD_{\hat w}$, which gives $(N\hat
c)_i= w_i^0 \hat c_j/\hat w_j$, where $i\in \alpha(j)$. With this and
$\hat w_j = \sum_{i\in \alpha(j)} w^0_i$ we obtain 
\begin{align*}
\hat\cE(\hat c)&=\cE_0(N\hat c) = \mathcal H_I(Nc|w^0)
 = \sum_i w_i^0 \,\LB\big(\frac{(N\hat c)_i}{ w^0_i}\big)\\ 
& = \sum_{j=1}^J \sum_{i\in \alpha(j)} w_i^0 \LB\big( 
          \frac{\hat c_j}{\hat w_j} \big)
= \sum_{j=1}^J \hat w_j \LB\big( 
          \frac{\hat c_j}{\hat w_j} \big) = \mathcal H_J(\hat c|\hat w).
\end{align*}  

\underline{Step 2: Legendre-conjugate pair $\hat\cR$ and $\hat \cR^*$.} We
start from the formula for $\hat\cR^*$ and calculate $\hat\cR$ as
follows. Using $MN=\mathrm{id}_Y$ and $\Xi=M^*Y^*$, we obtain
\begin{align*}
 &\hat \cR (\hat c, \hat v) 
 =\sup \bigset{ \la\hat\xi, MN\hat v\ra_J - 
       \hat \cR^*(\hat c, \hat \xi)}{\hat \xi \in Y^*} 
 \\
 & = \sup \bigset{ \la M^*\hat\xi, N\hat v\ra_I 
            -\hat \cR^*(\hat c, \hat \xi) }{\hat \xi \in Y^*} 
  = \sup \bigset{\la \xi, N\hat v\ra_I - \cR_S^*(N\hat
    c, \xi)}{\xi \in M^*Y^*} 
 \\
 & = \sup \bigset{\la\xi,N\hat v\ra_I -
    \cR_S^*(N\hat c, \xi) - \chi_\Xi(\xi) }{\xi \in X^*} =
  \cR_\eff(N\hat c, N \hat v),
\end{align*}
where we use the definition of $\cR^*_\eff$ in \eqref{eq:cR*eff}. 

\underline{Step 3: The gradient-flow equation for
  $(\hat\sfQ,\hat\cE,\hat\cR)$.} We first observe 
\begin{equation}
  \label{eq:DualProj}
  M^*N^* \D \cE_0(N\hat c) =  \D \cE_0(N\hat c).
\end{equation}
  Indeed, let us define the component-wise log-function on $\R^I$,
  $\log : x\mapsto (\log(x_i))_{i=1,\dots, I}$. We have $\D\cE_0(c) =
  \log(\DDD_{w^0}^{-1}c)$. Hence, for $c= N\hat c = \DDD_{w^0}M^*
  \DDD_{\hat w}^{-1}\hat c$, we conclude
\begin{align*}
  \D \cE_0(N\hat c) = \log(\DDD_{w^0}^{-1}N\hat c) = \log(M^*
  \DDD_{\hat w}^{-1}\hat c) = M^* \log( \DDD_{\hat w}^{-1}\hat c) =
  M^* \D \hat \cE(\hat c) = M^*N^* \D \cE_0(N\hat c),
\end{align*}
where we used that $\D \hat \cE(\hat c) = N^* \D \cE_0(N\hat c)$.

 With $\D\hat\cE(\hat c)=N^*\D\cE_0(N\hat c)$ and
\eqref{eq:DualProj} the gradient-flow equation for
$(\hat\sfQ,\hat\cE,\hat\cR)$ reads 
\begin{align*}
  \dot{\hat c} &= \partial_{\hat\xi}\hat\cR^*(\hat c,
  -\D\hat\cE_0(\hat c)) = M \partial_{\xi}\cR_S^*\big(N\hat c,
  -M^*\D\hat\cE(\hat c)\big) \\
&= M \partial_{\xi}\cR_S^*\big(N\hat c,
  -M^*N^*\D\cE_0(N\hat c)\big) = M \partial_{\xi}\cR_S^*(N\hat c,
  -\D\cE_0(N\hat c)) =M A^S N\hat c,
\end{align*}
 where we used the identity $\D_\xi
\cR^*_S(c,{-}\D\cE_0(c))=A^Sc$, which holds for all $c$ by the
construction of our gradient structure.  
\end{proof}

In analogy to formula \eqref{eq:hatA.j1j2} 
providing the coefficients $\hat A_{j_1j_2}$
of the coarse-grained generator $\hat A=MA^SN$ we can also give 
a formula for the tilting-invariant reaction intensities
$\kappa^{S,0}_{i_1i_2}$ to obtain the corresponding intensities
$\hat\kappa_{j_1,j_2}$ for the
coarse-grained equation \eqref{eq:CoarseGraEqn} by a suitable
averaging. In particular, the gradient systems 
$(\hat \sfQ,\hat \cE,\hat\cR)$ provides again a cosh gradient structure
in the sense of Section \ref{suu:Boltz.cosh}. 

\begin{prop}[Cosh structure of $\hat\cR{}^*$]\label{pr:kappa.hat.kappa}
The coarse-grained dual dissipation potential $\hat\cR{}^*$ reads 
\begin{align*}
&\hat\cR{}^*(\hat c,\hat \xi)=\hspace{-0.4em}
 \sum_{1\leq j_1<j_2\leq J} \hspace{-0.4em} \hat
\kappa_{j_1,j_2} \,\sqrt{\hat c_{j_1} \hat c_{j_2}} \,
 \sfC^* (\hat \xi_{j_1} {-} \hat\xi_{j_2} ) 
\text{ with }\hat\kappa_{j_1,j_2}= \!\!\sum_{i_1\in \alpha(j_1)} \sum_{i_2\in
  \alpha(j_2)}\!\! \kappa^{S,0}_{i_1i_2} \big( \tfrac{w^0_{i_1}w^0_{i_2}} 
 {\hat w_{j_1}\hat w_{j_2}} \big)^{1/2}. 
\end{align*}
\end{prop}
\begin{proof} Theorem \ref{thm:GS4CoarseGraEqn} provides an explicit
  formula for $\hat\cR{}^*$. Inserting the definitions of $M$ and $N$
  and grouping according to equivalence classes will provide the
  result. Recalling the function
  $\phi:\cI\to \cJ$ giving for each $i$ the associated equivalence
  class $\alpha(\phi(i))\subset \cI$ we have $(N\hat c)_i=w^0_i \hat
  c_{\phi(i)}/\hat w_{\phi(i)}$ and $(M^*\hat \xi)_i =
  \hat\xi_{\phi(i)}$ and find  
\begin{align*}
\hat\cR^*(\hat c,\hat \xi)&= \cR^*_S(N\hat c, M^*\hat \xi)= 
\frac12\sum_{i_1\in \cI}\sum_{i_2\in \cI} \kappa_{i_1 i_2}^{S,0}
\Big(\frac{w^0_{i_1} c_{\phi(i_1)}}{\hat w_{\phi(i_1)}}\,
   \frac{w^0_{i_2} c_{\phi(i_2)}}{\hat w_{\phi(i_2)}} \Big)^{1/2}
   \;\sfC^* \big(\hat\xi_{\phi(i_1)}- \hat \xi_{\phi(i_2)}\big)\\
&= \frac12\sum_{j_1\in \cJ}\sum_{j_2\in \cJ} 
     \sum_{i_1\in \alpha(i_1)}\sum_{i_2\in \alpha(j_2)} 
  \kappa_{i_1 i_2}^{S,0}
\Big(\frac{w^0_{i_1} c_{j_1}}{\hat w_{j_1}}\,
   \frac{w^0_{i_2} c_{j_2}}{\hat w_{j_2}} \Big)^{1/2}
   \;\sfC^* \big(\hat\xi_{j_1}{-} \hat \xi_{j_2}\big).  
\end{align*}
This shows the desired result. 
\end{proof}

\subsection{A priori bounds and compactness}
\label{su:BoundsComp}

We start the proof of the $\Gamma$-convergence for the dissipation
functional $\fD_\eps$ on $\rmL^2([0,T],\sfQ)$ by deriving the
necessary a priori bounds for proving the compactness  for a
family $(c^\eps)_{\eps>0}$ of functions satisfying $\fD_\eps(c^\eps)
\leq C < \infty$. 

Clearly since for all $t\in[0,T]$ we have $c^\eps(t)\in \sfQ$ we get
immediately uniform $\rmL^\infty$-bounds on $c^\eps$. Hence, we have
(after extracting a suitable subsequence, which is not relabeled) a
weak limit $c^0\in \rmL^2([0,T],\sfQ)$. We want to improve the
convergence to strong convergence.  Already in the proof of the
convergence of the solutions $c^\eps$ in Section \ref{su:ConvergSol}
it became clear that there are two different controls, namely (i) the 
tendency to go to microscopic equilibrium and (ii) the dissipation
through the slow reactions. From (i) we will obtain control
of the distance of $c^\eps$ from $X_\mathrm{eq}=PX$ by estimating
$(I{-}P)c^\eps$, but we are not able 
to control $(I{-}P)\dot c^\eps$. From (ii) we obtain an a priori bound for
$P\dot c^\eps$, and the major task is to show that these two
complementary pieces of information are enough to obtain compactness. 

Subsequently,
we will drop $\eps$ in the notations for $w^\eps$,
$\kappa^{\alpha,\eps}_{in}$, and $\cR_{S,\eps}$, and so on. Of course,
we will keep the important factor $1/\eps$    in
$\cR_\eps^*=\cR_S^*+\frac1\eps \cR^*_F$.

The following result shows the convergence of sequences to the
subspace $ X_\mathrm{eq} =PX$ of microscopic equilibria. Recall the
decomposition $X=X_\mathrm{eq} \oplus X_\mathrm{fast}$ from
\eqref{eq:DecompoX} and the projection $P=NM$ such that
$X_\mathrm{eq}=PX$ and $X_\mathrm{fast}=(I{-}P)X$. In particular, the
semi-norm $c\mapsto |(I{-}P)c|$ is equivalent to $c \mapsto
\mathrm{dist}(c,X_\mathrm{eq})$.

\begin{lem}[Convergence in the direction of fast reactions]
\label{le:CvgFastReactions}
Consider a sequence $(c^\eps)$ in $\rmL^2([0,T], \sfQ)$ with $\fD_\eps
(c^\eps)\leq C_\fD <\infty $ and $c^\eps\rightharpoonup c^0$ in
$\rmL^2([0,T], \R^I)$. Then, there is a constant $C>0$ such that
\[
\int_0^T |(I{-}P)c^\eps(t)|^2 \dd t \leq C\eps .
\]
In particular, we have $c^0(t)\in \sfQ_\mathrm{eq}=P\sfQ$ for a.a.\
$t\in [0,T]$.   
\end{lem}
\begin{proof} The bound on the dissipation functional $\fD_\eps$,
  $\cR_\eps\geq 0$, $\cR_{S}^*\geq 0$ and the relation $\sfC^*(\log p
  - \log q)=2\big(\sqrt{p/q} + \sqrt{q/p} -2)$ imply
\begin{align*}
C_\fD\geq \fD_\eps(c^\eps)\geq 
\frac 1 \eps \int_0^T  \sum_{(i,n)\in \mathbb F}
\frac{4 \kappa^F_{in}}{\sqrt{w_iw_n}}
\left(\sqrt{\frac{c^\eps_i}{w_i}}
  -\sqrt{\frac{c^\eps_n}{w_n}}\right)^2   \dd t , 
\end{align*}
where the set $\mathbb F$ is given in term of the equivalence relation
$\sim_F$, viz.
\[
\mathbb F:= \bigset{(i,n)\in \cI\ti\cI }{ i\sim_F n \text{ and } i<n}. 
\]
Using the decomposition $X=X_\mathrm{eq} \oplus X_\mathrm{fast}$ from
\eqref{eq:DecompoX}, we see that the semi-norm
\[
\|c\|_{\mathbb F} := \Big(\sum\nolimits_{(i,n)\in \mathbb F}
\big(\frac{c_i}{w_i} {-}\frac{c_n}{w_n}\big)^2 \Big)^{1/2} 
\]
defines a norm on $ X_\mathrm{fast}$ and there exists $C_2>0$
such that $|(I{-}P)c|\leq C_2 \| c\|_{\mathbb F} $ on $\sfQ$.

Denoting by $\ulw>0$ and $\ulk>0$ lower bounds for all $w^\eps_i$ and all
$\kappa^F_{in} $ with $i\sim_F n$, respectively, we obtain the
estimate 
\begin{align*}
&\int_0^T |(I{-}P)c^\eps(t)|^2 \dd t 
   \leq C_2^2\int_0^T \| c^\eps(t)\|_{\mathbb F}^2 \dd t\\
&\leq C_2^2\int_0^T \sum_{(i,n)\in \mathbb F}
\Big(\sqrt{\frac{c^\eps_i}{w_i}} {-}\sqrt{\frac{c^\eps_n}{w_n}}\Big)^2 
 \Big(\sqrt{\frac{c^\eps_i}{w_i}} {+}\sqrt{\frac{c^\eps_n}{w_n}}\Big)^2 \dd t\\
&\leq \frac{C_2^2}{\ulw^2\ulk} \int_0^T \sum_{(i,n)\in \mathbb
  F}\frac{4\kappa^F_{in}}{\sqrt{w_iw_n}} 
\Big(\sqrt{\frac{c^\eps_i}{w_i}} {-}\sqrt{\frac{c^\eps_n}{w_n}}\Big)^2
  \dd t \leq \frac{C_2^2}{\ulw^2\ulk}\: C_\fD\, \eps.
\end{align*} 

By weak lower semicontinuity of semi-norms we find
$\int_0^T|(I{-}P)c^0(t)|^2 \dd t =0$ and conclude $c^0(t)=Pc^0(t)$
a.e.\ on $[0,T]$. This proves the result. 
\end{proof}

The next result shows that we are able to control the time derivative
of $Pc^\eps$. Using $\mathrm{range}(P)=\mathrm{range}(N)$ and
$NM=\mathrm{id}_Y$ it suffices to control $M\dot c^\eps$. For this, we
show that $ \cR_\eps(c,\cdot) $ restricted to $PX$ has a uniform lower
superlinear bound in terms of the superlinear function $\sfC$, see
\eqref{eq:sfC.growth}.

\begin{prop}[Convergence in the direction of slow
  reactions]\label{PropRegularityofc} 
Consider a sequence $(c^\eps)$
in $\rmL^2([0,T], \sfQ)$ with $\fD_\eps (c^\eps)\leq C_\fD <\infty $ and 
$c^\eps\rightharpoonup c^0$ in $\rmL^2([0,T]; X)$. Then, there is a
constant $C_\rmW>0$ such that 
\begin{equation}
  \label{eq:Pdot.ceps}
  \int_0^T \sfC\big( \frac1{C_\rmW}\,|P\dot c^\eps(t)|\big) \dd t  \leq C_\rmW. 
\end{equation}
Moreover, $Pc^\eps \rightharpoonup Pc^0 \in \rmW^{1,1}([0,T];\sfQ)$
and $Pc^\eps \to Pc^0$ in $\rmC^0([0,T];P\sfQ)$. 

With Lemma \ref{le:CvgFastReactions} we have 
$c^\eps\rightarrow c^0$ strongly in $\rmL^2([0,T],\sfQ)$ and
$c^0=Pc^0\in \rmW^{1,1}([0,T];\sfQ)$. 
\end{prop}
\begin{proof} To show a lower bound for $\cR_\eps(c,Pv)$ we first
  derive an upper bound for $\cR_\eps^*(c, \tilde\xi)$ for $\tilde\xi\in
  P^*X^*$. Use $\cR^*_{F,\eps}(c,\tilde\xi)=0$ and set
  $\olka: = \sup\bigset{\kappa_{in}^{S,\eps}}{ 1\leq i<n\leq I, \eps
  \in {]0,1[}}$ to obtain
\begin{align*}
  \cR^*_\eps(c, \tilde\xi) 
& = \sum_{i<j}\kappa^{S,\eps}_{in} \sqrt{c_ic_j} \,\sfC^* 
     \big((\tilde\xi_i {-} \tilde\xi_j\big)
\leq \sum_{i<j}\olka \,\tfrac12 \,\sfC^*(\sqrt2\,|\tilde\xi| ) 
\leq  a\, \sfC^*(\sqrt2\,|\tilde\xi|)  
\end{align*}
with $ a=I^2\olka/4$. Next, Legendre transform,
$\cR^*_{F,\eps}(c,\tilde\xi)=0$ by (\ref{eq:cRF.Xi}) and the
bound $|P^*\xi|\leq C_P|\xi|/\sqrt{2}$ yield the lower bound
\begin{align*}
\cR_\eps(c,v) 
& \geq \sup\bigset{ \la \tilde\xi, v\ra {-} \cR^*_\eps(c, \tilde\xi)}{ 
 \tilde \xi\in P^*X^*}   =\sup\bigset{ \la P^*\xi, v\ra {-} 
   \cR^*_{S,\eps}(c, P^*\xi)}{\xi\in X^*} 
\\
& \geq \sup\bigset{ \la \xi, Pv\ra {-} a \sfC^*(\sqrt2
  |P^*\xi|)}{\xi\in X^*} \geq \sup\bigset{ \la \xi, Pv\ra {-} 
  a \sfC^*(C_P |\xi|)}{\xi\in X^*} \\
 & = a \,\sfC\big(|Pv|/(aC_P)\big). 
\end{align*}
Applying this to $v=\dot c^\eps$ we find 
\[
\int_0^T a \,\sfC \big(\frac{|P\dot c^\eps(t)|}{aC_P} \big) \dd t
\leq \int_0^T \cR_\eps(c^\eps(t),\dot c^\eps(t)) \dd t \leq
\fD_\eps(c^\eps) \leq C_\fD,
\]
which gives \eqref{eq:Pdot.ceps} with $C_\rmW=\max\{a C_P,
C_\fD/a\}$. \EEE

With the superlinearity of $\sfC$, we obtain $Pc^\eps \rightharpoonup Pc^0$ in
$\rmW^{1,1}([0,T];PX)$. Moreover, the sequence $Pc^\eps$ is also
equicontinuous, which is seen as follows. By \eqref{eq:Pdot.ceps} and
\eqref{eq:sfC.growth} we have $
\int_0^T |P\dot c^\eps(t)| \log\big(2{+}|P\dot c^\eps(t)|\big) \dd t
\leq C_1$. 
For $R>0$ we set $\Sigma(R,\eps)=\bigset{t\in [0,T]}{ |P\dot
  c^\eps(t)|\geq R}$. Thus, for $t_1<t_2$ we obtain the estimate 
\begin{align*}
&|Pc^\eps(t_2){-}Pc^\eps(t_1)| \leq \int_{t_1}^{t_2} |P\dot
c^\eps(t)|\dd t \\
&\leq \int_{[t_1,t_2]\setminus \Sigma(R,\eps)} \!\!|P\dot c^\eps(t)| \dd t
 + \int_{\Sigma(R,\eps)}\!\!|P\dot c^\eps(t)|\, \tfrac{ \log(2{+}|P\dot
  c^\eps(t)|)}{\log(2{+}R)} \dd t
\leq (t_2{-}t_1)R + \frac{C_1}{\log(2{+}R)}.
\end{align*}
The last sum can be made smaller than any $\eps>0$ by choosing first
$R=R(\eps):= \exp(2C_1/\eps)$ and then assuming $t_2-t_1< \delta
(\eps) := \eps/(2R(\eps))$. This shows $|P c^\eps (t_2) {-} P
c^\eps(t_1)|< \eps$ whenever $|t_2{-}t_1|< \delta(\eps)$, which is the
desired equicontinuity. By the Arzel\`a-Ascoli theorem we obtain
uniform convergence.

The final convergence follows from $c^\eps = Pc^\eps + (I{-}P)c^\eps$ via
Lemma \ref{le:CvgFastReactions}, and the last statement from
$Pc^0(t)=c^0(t)$ a.e. in $[0,T]$. 
\end{proof}

\subsection{The liminf estimate}
\label{su:LiminfEst}

For the limit passage $\eps\to 0$ we use a technique, which was
introduced formally in \cite{LMPR17MOGG} and exploited in
\cite{MaaMie18?MCRD} for the study of the large-volume limit in
chemical master equations. It relies on the idea that the velocity
part $\fD^\mathrm{vel}_\eps=\int \cR_\eps\,\d t$ \EEE of the
dissipation functional $\fD_\eps$ can be characterized by Legendre
transform using a classical result of Rockafellar:

\begin{thm}[{\cite[Thm.\,2]{Rock68ICF}}]
\label{thm:Rockafellar}
Let $f:[0,T]\ti \R^n\to \Rinfty$ be a normal, convex integrand and
with conjugate $f^*$. Assume there exist $u_\circ\in
\rmL^1([0,T];\R^n)$ and $\xi_\circ\in \rmL^\infty([0,T];\R^n)$ such
that $t\mapsto f(t,u_\circ(t))$ and $t\mapsto f^*(t,\eta_\circ(t))$
are integrable, then the functionals
\[
I_f{:}\left\{ \begin{array}{c@{\,}c@{\,}c}\!\!\rmL^1([0,T];\R^n)&\to& \Rinfty,\\
 u&\mapsto& \int_0^T \!f(t,u(t))\dd t \end{array}\right.
\text{ and } \ \ 
I_{f^*}:\left\{ \begin{array}{c@{\,}c@{\,}c}
  \!\!\rmL^\infty([0,T];\R^n)&\to
    &\Rinfty,\\
  \eta&\mapsto &\int_0^T \!f^*(t,\eta(t))\dd t \end{array}\right.
\] 
are proper convex functionals that are conjugate to each other with
respect to the dual pairing $(u,\eta)\mapsto \int_0^T\la \xi(t), 
u(t)\ra \dd t$, viz.\ for all $ u \in \rmL^1([0,T];\R^n)$ we have
\begin{equation}
  \label{eq:Rockafellar}
 \int_0^T \!\! f(t,u(t))\dd t = \sup \bigg\{\:
   \int_0^T \!\! \Big(\la\eta(t),u(t)\ra - f^*(t,\eta(t)) \Big) \dd t 
  \: \bigg|\:  \eta \in \rmL^\infty([0,T];\R^n) \: \bigg\}. 
\end{equation}
\end{thm}

We apply this result with $f(t,u)=\cR_\eps(c(t),u)$ and obtain, for
$\eps \in [0,1]$, the identity
\begin{align}
\label{eq:mfD=sup.mfB}
\fD_\eps(c)&= \sup\bigset{ \mfB_\eps(c,\dot c,\xi) }{
  \xi\in \EEE \rmL^\infty([0,T];X^*) }, \  \text{ where }
\\
\nonumber
& \mfB_\eps(c,u,\xi):=\mfB^\mathrm{vel}_\eps(
c,u,\xi)+\fD^\mathrm{slope}_\eps(c) \ \text{ with }
\\ 
\nonumber
&\mfB^\mathrm{vel}_\eps(c,u,\xi):=\int_0^T \Big(\la\xi(t),u(t)\ra -
\cR^*_\eps(c(t),\xi(t)) \Big)  \dd t \
\text{ and }
\\
&\nonumber
\fD^\mathrm{slope}_\eps(c) =
\int_0^T \cR^*_\eps\big(c,-\D\cE_\eps(c(t))\big) \EEE \dd t . 
\end{align}
The assumptions are easily satisfied as we may choose $u_\circ\equiv 0$ and
$\eta_\circ \equiv 0$.  

With these preparations we obtain the liminf estimate in a
straightforward manner.

\begin{thm}[Liminf estimate]
\label{TheoremLiminfEstimate}
The weak convergence  $c^\eps\rightharpoonup c^0$ in
$\rmL^2([0,T];\sfQ)$ implies 
 $\liminf_{\eps\rightarrow 0} \fD_\eps(c^\eps) \geq \fD_0(c^0)$, where
 $\fD_0$ is defined via $\cE_0$ and $\cR_\eff$ in \eqref{eq:frakD0}.
\end{thm}
\begin{proof} 
We may assume that $\alpha_*:=\liminf_{\eps\rightarrow 0}
\fD_\eps(c^\eps)<\infty$, since otherwise the desired estimate is
trivially satisfied.

\underline{\em Step 1. Strong convergence and limit characterization:}  
Using Proposition \ref{PropRegularityofc} gives
\[
c^\eps \to c^0 \text{ strongly in } \rmL^2([0,T];Q)\quad
\text{and} \quad c^0=Pc^0 \in \rmW^{1,1}([0,T];\R^I). 
\]

\underline{\em Step 2. Slope part:} Because of $Pc^0(t)=c^0(t)$ we
know $\xi_0(t)=\D\cE_0(c^0(t))\in M^*X^*$ which implies
$\chi_\Xi \big({-}\D\cE_0(c^0(t))\big)=0$ on $[0,T]$. Hence,
dropping the nonnegative term $\cR^*_{F,\eps}(c^\eps,
{-}\D\cE_\eps(c^\eps(t))\big)$ and setting $\mathcal S_\eps(c):=
\cR^*_{S,\eps}(c, {-}\D\cE_\eps(c))$ we obtain 
\[
\liminf_{\eps\to 0} \fD_\eps^\mathrm{slope}(c^\eps) 
\geq \liminf_{\eps\to 0} \int_0^T \mathcal S_\eps(c^\eps(t)) \dd t 
\overset*= 
\int_0^T \mathcal S_0(c^0(t))\dd t = \fD_0^\mathrm{slope}(c^0).
\]
In the passage $\overset*=$ we use the strong convergence $c^\eps
\to c^0$ and the continuity of 
\begin{equation}
  \label{eq:slopeSeps}
  [0,1]\ti Q\ni (\eps,c)\mapsto \mathcal S_\eps(c) = \cR^*_{S,\eps}(c,
{-}\D\cE_\eps(c))= \sum_{i<n} \frac{ 4 \kappa^{S,\eps}_{in}}{w_i^\eps
  w_n^\eps} \Big( \sqrt{\tfrac{c_i}{w_i^\eps}} -
\sqrt{\tfrac{c_n}{w_n^\eps}} \Big)^2. 
\end{equation}

\underline{\em Step 3. Velocity part:} We exploit the Rockafellar
representation \eqref{eq:mfD=sup.mfB} together with the fact that
$\dot c{}^0(t) = P\dot c{}^0(t)$ a.e.\ in $[0,T]$. The latter
condition allows us to test only by functions $\xi =P^*\xi \in 
\rmL^\infty([0,T];X^*)$, which leads to the estimate
\begin{align*}
\liminf_{\eps\to 0} \fD^\mathrm{vel}_\eps(c^\eps) &\geq \liminf_{\eps
  \to 0 } \mfB_\eps^\mathrm{vel}(c^\eps, \dot c^\eps, P^*\xi) 
\overset{\mathrm{a}}= \liminf_{\eps \to 0 } \int_0^T \!\!\Big(\la \xi,P\dot c^\eps \ra - 
\cR^*_{S,\eps}(c^\eps, \xi) \Big) \dd t\\
& \overset{\mathrm{b}}= \int_0^T\!\!\Big(\la \xi,P\dot c^0 \ra - 
\cR^*_{S}(c^0, \xi) \Big) \dd t= \mfB_0^\mathrm{vel}(c^0,\dot c{}^0,\xi),
\end{align*}
where in $\overset{\mathrm{a}}=$ we used
$\cR_\eps^*(c,\xi)=\cR^*_{S,\eps}(c,\xi)$ whenever $\xi=P^*\xi$, see
\eqref{eq:cRF.Xi}. In $\overset{\mathrm{b}}= $ we exploited the weak
convergence $P\dot c{}^\eps \rightharpoonup P c{}^0$ established in
Proposition \ref{PropRegularityofc} as well as the strong convergence
$c^\eps \to c^0$ together with the continuity of $(\eps,c)\mapsto
\cR^*_{S,\eps}(c,\xi)$.

Now we exploit Rockafellar's characterization \eqref{eq:mfD=sup.mfB}
to return to $\fD^\mathrm{vel}_0(c^0)$, namely
\begin{align*}
 \fD^\mathrm{vel}_0(c^0)&=\sup\bigset{\mfB_0^\mathrm{vel}(c^0,
   \dot c^0, \xi) }{ \xi\in\rmL^\infty([0,T];X^*)}\\
&=\sup\bigset{\int_0^T\Big(\la \xi,\dot c^0\ra -
  \cR^*_S(c^0,\xi)-\chi_\Xi(\xi) \Big) \dd t }{
  \xi \in \rmL^\infty([0,T];X^*)}
\\
&=\sup\bigset{ \mfB_0^\mathrm{vel}(c^0,\dot c{}^0,\xi) }{
  \xi= P^* \xi \EEE \in \rmL^\infty([0,T];X^*)}.
\end{align*}
With the above estimate we conclude $\liminf_{\eps\to 0}
\fD^\mathrm{vel}_\eps(c^\eps) \geq \fD^\mathrm{vel}_0(c^0)$.

Adding this to the estimate in Step 2 we obtain the full liminf estimate. 
\end{proof}

\subsection{Construction of the recovery sequence}
\label{su:RecoverySeq}

Now we construct the recovery sequence for the Mosco-convergence of
the dissipation functionals $\fD_\eps$. This provides the required
limsup estimate $\limsup_{\eps \to 0 } \fD_\eps(c^\eps) \leq
\fD_0(c^0)$ along at least one sequence with the strong convergence
$c^\eps \to c^0$ in $\rmL^2([0,T];\sfQ)$.  For this we use in Step
2(b) an approximation result by piecewise affine functions $\widehat
c_N$ introduced in \cite[Thm.\,2.6, Step\,3]{LieRei18HCHT} and
adapted to state-dependent dissipation potentials in
\cite[Cor.\,3.3]{BaEmMi18?EREG}.

\begin{thm}[Recovery sequences]
  For every $c^0\in\rmL^2([0,T];\sfQ)$ there exists a
  sequence $(c^\eps)_{\eps \in {]0,1[}}$ with $c^\eps\rightarrow c^0$ in
  $\rmL^2([0,T];\sfQ)$ such that $\lim_{\eps\rightarrow
    0}\fD_\eps(c^\eps)=\fD_0(c^0)$.
\end{thm}
\begin{proof} \underline{\em Step 1. The case $\fD_0(c^0)=\infty$.}
  We choose the constant sequence $c^\eps=c^0$ and claim
  $\fD_\eps(c^\eps)=\fD_\eps(c^0)\to \infty$.  Because of
  $\fD_0(c^0)=\infty$ one of the following conditions is false:
\medskip 

\centerline{(i)
  $c^{0}(t)\in \sfQ_\mathrm{eq}$ a.e.\ in $[0,T]$ \quad or \quad (ii)
  $\sfC\big(|P\dot c{}^0(\cdot)|\big) \in \rmL^1([0,T])$.\medskip }

If (i) is false, then $c^0(t)\not\in \sfQ_\mathrm{eq}$ for $t\in \mathcal
T \subset[0,T]$, where $|\mathcal T|=\int_{\mathcal T}1 \dd t >0$. 
Setting $\mathcal F_\eps(c):=\cR^*_{F,\eps}(c,{-}\D\cE_\eps(c))$ we
have 
\[
\fD^\mathrm{slope}_\eps(c^0)=
\int_0^T\!\!\Big(\cR_{S,\eps}^*\big(c^0,{-}\D\cE_\eps(c^0) \big){+}\frac1\eps 
 \cR^*_{F,\eps}\big(c^0,{-}\D\cE_\eps(c^0) \big) \Big) \dd t 
 \geq \frac1\eps \int_0^T\mathcal F_\eps(c^0(t)) \dd t.  
\]
However, for $t\in \mathcal T$ we have $\mathcal F_\eps(c^0(t))\to
\mathcal F_0(c^0(t)) >0$. Thus, $\fD^\mathrm{slope}_\eps(c^0)\to
\infty$ follows which implies $\fD_\eps(c^0)\to \infty$.

If (ii) is false, then $\fD_\eps^\mathrm{vel}(c^0)=\infty$ for all
$\eps>0$ and we are done.\medskip

\underline{\em Step 2. Preliminary recovery sequences for the case
  $\fD_0(c^0)<\infty$.} In the sub-steps 
(a) to (c) we discuss three approximations for general
$c^0$.\smallskip

{\em Step 2(a).} Positivity for the case $\eps=0$. We set $\tilde
c_\delta(t):=\delta w^0+ (1{-}\delta) c^0(t)$ and claim that
$\fD_0(\tilde c_\delta ) \to \fD_0(c^0)<\infty$ for $\delta\searrow
0$. As $\fD_0$ is convex and lower semicontinuous (cf.\ see Proposition
\ref{pr:Proper.cReps}), we have
$\liminf_{\delta \searrow 0} \fD_0(\tilde c_\delta) \geq \cD_0(c^0)$.

Obviously, $\tilde c_\delta \geq (1{-}\delta) c^0$ holds
componentwise, and hence the explicit form of $\cR^*_0$ gives
\[
\cR^*_\eff(\tilde
c_\delta, \xi)\geq (1{-}\delta) \cR^*_\eff(c^0,\xi),  
\quad\text{and thus  }
\cR_\eff(\tilde c_\delta,v)\leq  (1{-}\delta) \cR_\eff \big(c^0,
\frac1{1{-}\delta} v\big).
\]
Inserting $v=\dot{\tilde c}_\delta= (1{-}\delta) \dot c{}^0$ into the
latter estimate gives 
\[
\fD_0^\mathrm{vel}(\tilde c_\delta)=\int_0^T \!\!\cR_\eff(\tilde c_\delta, 
 \dot{\tilde c}_\delta) \dd t \leq
 \int_0^T\!\!(1{-}\delta)\cR_\eff(c^0,\dot c{}^0) \dd t = (1{-}\delta)
 \fD_0^\mathrm{vel}(c^0),
\]
which proves the desired claim of Step 2(a), because
$\fD_0^\mathrm{slope}( \tilde c_\delta) \to \fD_0^\mathrm{slope}(c^0)$
is trivial.\smallskip

{\em Step 2(b).} We stay with $\eps=0$ and, by Step 2(a), may
assume for some $c_*>0$ that
\[
c^0(t) \in \sfQ_{c_*}:=\bigset{c\in \sfQ}{\forall\,i\in \cI{:}\
  c_i\geq c_*}  \text{ for all } t \in [0,T].
\]
We now approximate $c^0$ by
a function $\widehat c_N \in \rmW^{1,\infty}([0,T];PX)$ still
satisfying $\widehat c_N(t)\in \sfQ_{c_*}$. 

For $N\in \N$ we define $\widehat c_N:[0,T]\to PX$ as the piecewise
affine interpolant of the nodal points $\widehat c_N(kT/N)=c^0(kT/N)$
for $k=0,1,\ldots, N$.  We also define the piecewise constant
interpolant $\overline c_N:[0,T]\to \sfQ_{c_*}$ via $\overline
c_N(t)=c^0(kT/N)$ for $t\in {](k{-}1)T/N,kT/N]}$.  Then, using $c^0\in
\rmW^{1,1}([0,T];PX) \subset \rmC^0([0,T];PX)$ we have
\[
\widehat c_N\to c^0 \text{ in }\rmW^{1,1}([0,T];PX) \text{ and in } 
\rmC^0([0,T];PX)\quad \text{and} \quad 
\overline c_N \to c^0 \text{ in } \rmL^\infty([0,T];PX).  
\]
We now set
\[\alpha_N:=\|c^0{-}\widehat c_N\|_{\rmL^\infty} + \| c^0{-}\overline
c_N\|_{\rmL^\infty} 
\] 
and obtain $\alpha_N\to 0$.

These uniform estimates can be used in conjunction with the uniform
continuity of $c\mapsto \cR^*_\eff(c,\xi)$ when restricted to
$\sfQ_{c_*}$. Clearly $\sfQ_{c_*}\ni c\mapsto \sqrt{c_ic_n}$ is
Lipschitz continuous, and we call the Lipschitz constant $\lambda^*$.
The special form of $\cR^*_\eff$ then implies 
\[
\forall\,c,\tilde c\in \sfQ_{c_*} \ \forall\, \xi\in X^*:\ 
|\cR^*_\eff(c,\xi)-\cR^*_\eff(\tilde c,\xi)| \leq \Lambda^*
|c{-}\tilde c| \cR^*_\eff(c,\xi) \ \text{ with
}\Lambda^*=\lambda^*\overline \kappa.
\]
Assuming $|c{-}\tilde c|\leq \alpha$ and $\Lambda^*\alpha<1$ and
applying the Legendre transform we find 
\[
(1{-}\Lambda ^*\alpha)\cR_\eff(c,\tfrac1{1{-}\Lambda^*\alpha}v) 
 \geq \cR_\eff(\tilde c,v) \geq 
(1{+}\Lambda ^*\alpha)\cR_\eff(c,\tfrac1{1{+}\Lambda^*\alpha}v) .
\] 
Exploiting the scaling property \eqref{eq:Rcosh.scaling} we arrive at
the estimates 
\[
\frac1{1{-}\Lambda ^*\alpha}\,\cR_\eff(c, v ) 
 \geq \cR_\eff(\tilde c,v) \geq 
\frac1{1{+}\Lambda ^*\alpha}\, \cR_\eff(c, v) .
\] 

To estimate the velocity part of the dissipation functional as in
\cite{LieRei18HCHT, BaEmMi18?EREG} we introduce
\[
\mathfrak J(c,v):=\int_0^T \cR_\eff(c(t),v(t))\dd t,
\]
which allows us to use different approximations for $c^0$ and for $\dot
c^0$. We obtain
\begin{align*}
\fD^\mathrm{vel}_0(\widehat c_N)&= \mathfrak J(\widehat c_N,
\dot{\widehat c}_N) \leq (1{+}\Lambda^*\alpha_N) \mathfrak J(\overline
c_N, \dot{\widehat c}_N) \\
&\overset*\leq (1{+}\Lambda^*\alpha_N) \mathfrak J(\overline c_N, \dot c^0) 
\leq (1{+}\Lambda^*\alpha_N)^2 \mathfrak J(c^0,\dot c^0) =
 (1{+}\Lambda^*\alpha_N)^2 \fD^\mathrm{vel}_0( c^0).
\end{align*}
For the estimate $\overset*\leq$ we split $[0,T]$ into the
subintervals $S_k^N:={](k{-}1)T/N,kT/N[}$, where $\overline c_N$ and
$\dot{\widehat c}_N$ are equal to the constants $c^0(kT/N)$ and $\frac
TN\int_{S_N} \dot c^0(t) \dd t$, respectively. Then, Jensen's
inequality for the convex function $\cR_\eff(\overline
c_N,\cdot)$ \EEE
gives the desired estimate.  

Since $\fD^\text{slope}_0(\widehat c_N)\to \fD^\text{slope}_0(c^0)$ by
the continuity of the integrand $\mathcal S_0$ (cf.\
\eqref{eq:slopeSeps}), the lower semicontinuity of $\fD_0$ yields
$\fD_0(\widehat c_N)\to \fD_0(c^0)$. \smallskip

{\em Step 2(c).} Using the Step 2(a) and 2(b), we now may assume 
$c^0\in \rmW^{1,\infty}([0,T],X)$ with $c^0(t)\in \sfQ_{c_*}$ on
$[0,T]$ and define $c^\eps$ via the formula
\[
c^\eps(t) = \DDD_{w^\eps}\DDD_{w^0}^{-1} c^0(t) \text{ for } t \in
[0,T].
\]
This definition gives $\D \cE_\eps(c^\eps(t))\in \Xi$ and hence
$\mathcal S^F_\eps(c^\eps(t))=0$. Hence, the
definition of $\mathcal S_\eps$ in terms of the ratios $c_i/w^\eps_i$
(cf.\ \eqref{eq:slopeSeps}) implies $\fD^\text{slope}_\eps(c^\eps)
\to \fD_0^\text{slope}(c^0)$. 

For the velocity part we again use the Rockafellar characterization,
namely
\[
\fD^\text{vel}_\eps(c^\eps) = \sup \bigset{ \mathfrak
  B^\mathrm{vel}_\eps( c^\eps, \dot c^\eps, \xi) }{ \xi \in
  \rmL^\infty([0,T];X^*)  }.  
\]
Because of the uniform bound for $\dot c^\eps$ in
$\rmL^\infty([0,T];X)$ we are able to show \EEE that the supremum
over $\mathfrak B^\mathrm{vel}_\eps(c^\eps,\dot c^\eps,\cdot)$ is
attained by maximizers $\xi^\eps$ that are uniformly bounded in
$\rmL^\infty([0,T];X^*)$.  To see this, we firstly observe that
for all $(c,\dot c)\in\sfQ\times X$ the functional
$\mathfrak{B}_\eps(c,\dot c,\cdot)$ is invariant under translation
$\xi\mapsto \xi+c\one$, since $\la\one,\dot{c}\ra=0$ and the
dissipation potential $\cR^*_\eps(c,\xi)$ only depends on differences
$\xi_i-\xi_n$. Hence, to generate compactness for maximizers $\xi^\eps$
we fix the first component by setting$\xi^\eps_1=0$. Then, by
$c^0(t)\in \sfQ_{c_*}$ and the exponential growth of $\cR^*_\eps$ we
obtain
\[
\cR^*_\eps(c^\eps(t),\xi) \geq c_\circ \sum_{i<n:A_{in}^\eps>0}|\xi_i-\xi_n|^2,
\]
with a positive constant $c_\circ>0$. By assumption \eqref{eq:cond.A1}
on the connectivity of the reaction network, all
vertices can be reached by a reaction path from vertex 1. Hence, there
is another constant $\tilde c_\circ$, such that the estimate
\[
\cR^*_\eps(c^\eps(t),\xi) \ \geq \ c_\circ\! \sum_{i<n:A_{in}^\eps>0} 
 |\xi_i-\xi_n|^2 \ \geq \ \tilde c_\circ \sum_{i>1}|\xi_1-\xi_i|^2
\]
holds.
Hence, the maximizers $\xi^\eps$ for $\fD^\text{vel}_\eps(c^\eps)$ 
with $\xi^\eps_1=0$ satisfy
\[
\la \xi^\eps,\dot c^\eps(t)\ra- 
\cR^*_\eps(c^\eps(t),\xi) \leq |\xi^\eps| C \|\dot c^0\|_{\rmL^\infty([0,T],X)}- \tilde c_\circ |\xi^\eps|^2,
\]
which implies the uniform bound $\|\xi^\eps\|_{\rmL^\infty([0,T],X^*)} \leq C \|\dot
c^0\|_{\rmL^\infty([0,T],X)}/\tilde c_\circ$. \EEE

We now first choose a subsequence $(\eps_k)_{k\in \N}$  such that
$\eps_k \searrow 0$ and 
$\fD^\mathrm{vel}_{\eps_k}( c^{\eps_k}) \to \beta=\limsup_{\eps\to 0} 
\fD^\mathrm{vel}_{\eps}( c^{\eps})$.  Thus, after selecting a further 
subsequence (not relabeled) we may assume $\xi^{\eps_k} \rightharpoonup
\xi^0$ in $\rmL^2([0,T];X^*)$. With the strong convergence of $\dot
c^\eps \to \dot c^0$ we conclude 
\[
\limsup_{\eps \to 0} \fD^\mathrm{vel}_\eps(c^\eps) =
\lim_{k\to \infty} \mathfrak B^\mathrm{vel}_{\eps_k}( c^{\eps_k},
\dot c^{\eps_k}, \xi^{\eps_k}) 
\overset*\leq 
\mathfrak B^\mathrm{vel}_0 (c^0,\dot c^0,\xi^0) \leq \fD^\mathrm{vel}_0(c^0),  
\] 
where in $\overset*\leq$ we used the convergence of the duality
pairing $\int_0^T \la \xi^\eps,\dot c^\eps\ra \dd t$ and a Ioffe-type
argument based on the convexity of $\cR^*_\eps(c^\eps,\cdot)$ and the
lower semicontinuity of $[0,1]\ti X^*\ni (\eps,\xi)\mapsto
\cR^*_\eps(c^\eps(t),\xi)\in [0,\infty]$, cf.\
\cite[Thm.\,7.5]{FonLeo07MMCV}. Adding the convergence of the slope
part, and taking into account the liminf estimate from Theorem
\ref{TheoremLiminfEstimate} we obtain the convergence $\lim_{\eps \to
  0} \fD_\eps(c^\eps) = \fD_0(c^0)$.  \medskip

\underline{\em Step 3. Construction of recovery sequences for the case
  $\fD_0(c^0)$.} We now
apply the approximation steps discussed in Step 2 and show that it is
possible to choose an suitable diagonal sequence for getting the
desired recovery sequence. 

For a  general $c^0$ we apply the approximation as indicated in the
sub-steps 2(a), 2(b), and 2(c) and set
\[
c^{\delta,N,\eps} = \DDD_{w^\eps}^{-1} \DDD_{w_0} \big( \delta w^0
+(1{-}\delta) \widehat c{}^{\;0}_N\big) .
\]
We easily obtain
$\|c^0 -c^{\delta,N,\eps}\|_{\rmL^2([0,T];X)} \leq C(\delta +
\alpha_N+\eps) \to 0$ if $\delta\to 0$, $N\to \infty$, and
$\eps \to 0$.  Moreover, the difference in the dissipation
functionals $\fD_\eps$ can be estimated via
\begin{align*}
\mbox{}\quad\big|\fD_\eps(c^{\delta,N,\eps}) - \fD_0(c^0)\big| &\leq A(\delta)+
B_\delta(N)+ C_{\delta,N}(\eps),&&  \text{where } \\
A(\delta)&=\big|\fD_0(\tilde c^{\:\delta}) - \fD_0(c^0)\big|  
 && \text{with } \tilde c^{\:\delta}(t)=\delta w^0 +(1{-}\delta)c^0(t), \\
B_\delta(N)&= \big|\fD_0\big( c^{\delta,N}\big) -
\fD_0(\tilde c^{\:\delta})\big| &&\text{with } c^{\delta,N}(t)=\delta w^0
+(1{-}\delta) \widehat c{}^{\;0}_N(t),\quad \mbox{}\\
C_{\delta,N}(\eps)&=\big|\fD_\eps(c^{\delta,N,\eps}) -
\fD_0(c^{\delta,N}) \big|.
\end{align*}
A recovery sequence $c^\eps \to c^0$ with $\fD_\eps(c^\eps) \to
\fD_0(c^0)$ is now obtained by a standard diagonal argument. \EEE
\end{proof}

\appendix

\section{Special properties of cosh gradient structures}
\label{se:ApropC}

Here we discuss a few special properties that are characterizing for
the function $\sfC$ and $\sfC^*$ and this lead to corresponding
properties of $\cR^*_\mathrm{cosh}$. 

We consider the special non-quadratic dissipation functional
\begin{align*}
 \sfC(v): = 2v~\mathrm{arsinh}(v/2) - 2\sqrt{4{+}v^2} + 4  \quad 
\text{and \quad its Legendre dual } \ \sfC^*(\xi) := 4\cosh(\xi/2)-4.
\end{align*}
Then, we have $\sfC(v) = \frac 1 2 v^2 + O(v^4)$ and $\mathsf
C^*(\xi) = \frac 1 2 \xi ^2 + O(\xi ^4)$.  The function $\sfC^*$
has the following properties:
\begin{align}
 \label{eq:A.prop.sfC}
  \sfC^*(\log p - \log q) = 2 \Big(\sqrt{\frac pq} + \sqrt{\frac qp}
  -2 \Big), \quad \sfC^{*\prime}(\log p - \log q) =
  \frac{p-q}{\sqrt{pq}}.
\end{align}

In addition we have superlinear growth of $\sfC$:
\begin{equation}
  \label{eq:sfC.growth}
\frac12 |s|\log(1{+}|s|) \leq \sfC(s) \leq   2 |s|\log(1{+}|s|) \EEE \quad
  \text{for all } s \in \R. 
\end{equation}
 
The first of the following scaling properties follows easily by
considering the power series expansion of $\sfC^*$, the second by
Legendre transform:
\begin{equation}
  \label{eq:scC.scaling}
  \forall\,\lambda \geq 1\ \forall \, s,\zeta\in \R:\quad
       \sfC^*(\lambda \zeta) \geq \lambda^2 \sfC^*(\zeta)
  \ \text{ and } \  
    \sfC(\lambda s) \leq \lambda^2 \sfC(s).
\end{equation}
This implies the corresponding scaling property for $\cR_\mathrm{cosh}$, namely
\begin{equation}
  \label{eq:Rcosh.scaling}
\begin{aligned} 
& \forall\, \lambda\geq 1\ \forall\,c\in\sfQ\ \forall\, v,\xi\in \R^I:
\\ 
& \qquad       \cR^*_\mathrm{cosh}(c,\lambda \xi) 
          \geq \lambda^2 \cR^*_\mathrm{cosh}(c,\xi)
   \text{ and } 
    \cR_\mathrm{cosh}(c,\lambda v) \leq 
   \lambda^2 \cR_\mathrm{cosh}(c,v).
\end{aligned} 
\end{equation}

\paragraph*{Acknowledgments.} 
This research has been partially supported by Deutsche
Forschungsgemeinschaft (DFG) through the Collaborative Research Center
SFB 1114 "Scaling Cascades in Complex Systems" (Project Number
235221301), subproject C05 ``Effective Models for Materials and
Interfaces with Multiple Scales''. The authors are grateful to 
two anonymous referees
for helpful comments and to Mark
Peletier for stimulating discussions  and for pointing out the wrong
formulation of the previous version of Lemma \ref{le:EDPimpliesCvg}. \EEE

\footnotesize


\newcommand{\etalchar}[1]{$^{#1}$}
\def\cprime{$'$}

\end{document}